\documentclass[final]{siamltex}

\usepackage{amsmath,amssymb} 
\usepackage{mathrsfs}
\usepackage{graphicx}
\begin{document}

\title{Iterative rounding approximation algorithms for degree-bounded
node-connectivity network design\thanks{A preliminary version of this paper appeared in proceedings of the 53rd Annual IEEE Symposium on Foundations of Computer Science (FOCS 2012).}}

\author{Takuro Fukunaga\thanks{National Institute of Informatics,
Japan. JST, ERATO, Kawarabayashi Large Graph Project} \and Zeev Nutov\thanks{Open University of Israel} \and R.~Ravi\thanks{Carnegie Mellon University}}

\maketitle

\newcommand{\Qset}{\mathbb{Q}}
\newcommand{\Rset}{\mathbb{R}}
\newcommand{\Zset}{\mathbb{Z}}
\newcommand{\Vset}{\mathcal{V}}
\newcommand{\Fset}{\mathcal{F}}
\newcommand{\Lset}{\mathcal{L}}
\newcommand{\Mset}{\mathcal{M}}
\newcommand{\Eset}{\mathcal{E}}
\renewcommand{\span}{{\rm span}}
\newcommand{\comp}[1]{\overline{#1}}

\newcommand{\SN}{{\sf Survivable Network}}
\newcommand{\DB}{{\sf Degree-Bounded}}

\newcommand{\EC}{{\sf  Edge-Connectivity}}
\newcommand{\ElC}{{\sf Element-Connectivity}}
\newcommand{\NC}{{\sf  Node-Connectivity}}

\newcommand{\kCS}{{\sf  $k$-Connected Subgraph}}
\newcommand{\kOS}{{\sf  $k$-Out-connected Subgraph}}
\newcommand{\SkCS}{{\sf Subset $k$-Connected Subgraph}}

\newcommand{\fCS}{{\sf $f$-Connected Subgraph}}

\begin{abstract}
We consider the problem of finding a minimum edge cost subgraph of a graph satisfying both
given node-connectivity requirements and degree upper bounds on nodes.  We present an
iterative rounding algorithm of the biset LP relaxation for this problem.
For directed graphs and $k$-out-connectivity requirements from a root, 
our algorithm computes a solution that is a 2-approximation on the cost, and 
the degree of each node $v$ in the solution is at most $2b(v) + O(k)$
 where $b(v)$ is the degree upper bound on $v$.
For undirected graphs and element-connectivity requirements with maximum connectivity requirement $k$,
our algorithm computes a solution
that is a $4$-approximation on the cost, and 
the degree of each node $v$ in the solution is at most $4b(v)+O(k)$.
These ratios improve the previous $O(\log k)$-approximation on the cost and
$O(2^k b(v))$ approximation on the degrees.
Our algorithms can be used to improve approximation ratios
for other node-connectivity problems such as
undirected $k$-out-connectivity, directed and undirected $k$-connectivity,
and undirected rooted $k$-connectivity and subset $k$-connectivity.
\end{abstract}

\section{Introduction} \label{s:intro}

\subsection{Problem definition} \label{ss:problems}

We consider the problem of finding a minimum edge cost subgraph that satisfies
both degree-bounds on nodes and certain connectivity requirements between nodes.
More formally, the problem is defined as follows.

\vspace*{0.2cm}

\begin{center}
\fbox{
\begin{minipage}{0.960\textwidth}
\noindent
{\bf \em Degree-Bounded Survivable Network} \\
A directed/undirected graph $G=(V,E)$ with edge costs $c\colon E \rightarrow \Rset_+$,
connectivity requirements $r\colon V \times V \rightarrow \Zset_+$,
and degree-bounds $b\colon B \rightarrow \Zset_+$ on a subset $B$ of $V$ are given.
The goal is to find a minimum cost edge set $F \subseteq E$ such that in the 
subgraph $(V,F)$ of $G$, the $uv$-connectivity is at least $r(u,v)$ for any $(u,v) \in V \times V$, 
and the out-degree/degree of each $v \in B$ is at most $b(v)$.
\end{minipage}
}
\end{center}

\vspace*{0.2cm}

In the case of digraphs, our algorithms easily extend to the case
when we also have {\em in-degree} bounds $b^-\colon B^-\rightarrow \Zset_+$ on $B^- \subseteq V$,
and require that the in-degree of each $v \in B^-$ is at most $b^-(v)$.

A node $v \in V$ is called {\em terminal} if there exists
$u \in V \setminus \{v\}$ such that $r(u,v) > 0$ or $r(v,u) > 0$.
We let $T$ denote the set of terminals.
We denote $\max_{u,v \in V}r(u,v)$ by $k$ and $|V|$ by $n$ throughout the paper.

If  $G$ is an undirected graph, $B=V$, $b(v)=2$ for all $v \in B$,
and a solution is required to be a connected spanning subgraph, then we get
the {\sf Hamiltonian Path} problem, and hence it is NP-hard even to find
a feasible solution. Therefore we consider bicriteria
approximations by relaxing the degree-bounds.
We say that an algorithm for {\DB} {\SN} is
$(\alpha, \beta(b(v)))$-approximation, or that it has ratio $(\alpha, \beta(b(v)))$,
for $\alpha \in \Rset_+$ and a function $\beta\colon  \Zset_+\rightarrow \Rset_+$,
if it always outputs a solution such that its cost is at most $\alpha$ times the optimal value,
and the degree  (the out-degree, in the case of directed graphs)
of each $v \in B$ is at most $\beta(b(v))$,
for any instance which admits a feasible solution.

Notice that {\DB} {\SN} includes the problem of
finding a minimum cost subgraph of required connectivity minimizing the maximum
degree.  This can be done by letting $B=V$, and defining $b(v)$ as the
uniform bound on the optimal value for all $v \in B$, which can be
computed by binary search.
For LP based algorithms, such as ours,
when a lower bound is obtained by solving an LP relaxation, 
these algorithms also establish a ``relaxed'' integrality gap for
the LP relaxation (relaxed since the solutions violate the exact degree requirements).

In this paper, we are interested in node-connec\-tivity and element-connectivity requirements.
The node-connectivity $\kappa(u,v)$ is the maximum number of $(u,v)$-paths
that are pair-wise internally (node) disjoint.
The definition of element-connectivity assumes that a terminal set
$T$ is given. The element-connectivity $\lambda^T(u,v)$
between two terminals $u,v \in T$ is the maximum number of $(u,v)$-paths
that are pair-wise disjoint in edges and in non-terminal nodes.
The two main problems we consider are as follows.

\vspace*{0.2cm}

\begin{center}
\fbox{
\begin{minipage}{0.960\textwidth}
\noindent
{\bf \em Degree-Bounded $k$-Out-connected Subgraph} \\
This is a particular case of {\DB} {\SN} when it is required that $(V,F)$ 
is $k$-outconnected from a given root $s$,
namely, when $\kappa(s,v) \geq k$ for each $v \in V\setminus \{s\}$.
\end{minipage}
}
\end{center}

\vspace*{0.2cm}

\begin{center}
\fbox{
\begin{minipage}{0.960\textwidth}
\noindent
{\bf \em Degree-Bounded Element-Connectivity Survivable Network} \\
This is a particular case of {\DB} {\SN} when
the input graph is {\em undirected}, and it is required that
$\lambda^T(u,v) \geq r(u,v)$ for each $u,v\in T$, where $T \subseteq V$ is a given terminal set.
\end{minipage}
}
\end{center}

\vspace*{0.2cm}

Our main results for {\DB} {\kOS} is for directed graphs. 
We also present similar results for undirected graphs,
but these are derived from the ones for the directed case.
Another fundamental problem that we consider for both directed and undirected graphs is as follows.

\vspace*{0.2cm}

\begin{center}
\fbox{
\begin{minipage}{0.960\textwidth}
\noindent
{\bf \em Degree-Bounded $k$-Connected Subgraph} \\
This is a particular case of {\DB} {\SN} where
it is required that $(V,F)$ is $k$-connected, namely, that $\kappa(u,v) \geq k$ for each $u,v\in V$.
\end{minipage}
}
\end{center}

\vspace*{0.2cm}

Other special cases of {\DB} {\SN} on undirected graphs are defined
according to the connectivity requirements, as follows.

\begin{list}{\textbullet}{\leftmargin=1.5em}
 \item
       {\bf {\em Degree-Bounded Node-Connectivity Survivable Network}}
       requires $\kappa(u,v) \geq r(u,v)$ for each $u,v\in V$.
 \item {\bf {\em Degree-Bounded Rooted Survivable Network}}
       is a special case of {\em Degree-Bounded Node-Connectivity Survivable Network}
       where a {\em root} node $s$ is specified,
       and $r(u,v)=0$ holds if $u \neq s$ and $v \neq s$.
 \item {\bf {\em Degree-Bounded Subset $k$-Connected Subgraph}}
       is a special case of {\em Degree-Bounded Node-Connectivity Survivable Network}
       where $r(u,v)=k$ if $\{u,v\} \subseteq T$, and $r(u,v)=0$ otherwise,
       for a given terminal set $T\subseteq V$.
\end{list}

\subsection{Previous work} \label{ss:prev-work}

\subsubsection{{\EC} {\SN}}

{\SN} without degree-bounds is a typical combinatorial optimization
problem, that was studied extensively; c.f. \cite{Frank2011} for a survey on exact algorithms, and 
\cite{Kortsarz2007} for a survey on approximation algorithms for
various {\SN} problems and their classification  w.r.t. costs and requirements.
One of the most important methods for these problems is iterative rounding, that was invented in the
context of a $2$-approximation algorithm by Jain~\cite{Jain2001}.  He
showed that every basic optimal solution to an LP relaxation for the
undirected {\EC} {\SN} always has a variable of value at least $1/2$;
see \cite{Nagarajan2010} for a simplified proof.
The $2$-approximation algorithm is obtained by repeatedly rounding up such
variables and iterating the procedure until the rounded subgraph is feasible.

{\DB} {\SN}, even with edge-connectivity requirements, was regarded as a difficult problem
for a long time because of the above-mentioned hardness on feasibility.
A breakthrough was given by Lau, Naor, Salavatipour, and
Singh~\cite{Lau2009} and Singh and Lau~\cite{Singh2007}.  They gave a
$(2,2b(v)+3)$-approxima\-ti\-on for the {\sf Degree-Bounded Edge-Connectivity \SN} problem, and a
$(1,b(v)+1)$-approximation algorithm for the {\sf Degree-Bounded Spanning Tree} problem.
The former result was improved (for large $b(v)$) to a $(2,b(v)+6k+3)$-approximation
by Lau and Singh~\cite{Lau2008} afterwards, which was followed by
further improvements
by Louis and Vishnoi~\cite{LouisV2010} and 
Lau and Zhou~\cite{LauZ2014}.

After their work, many efficient algorithms have been
proposed for various types of {\DB} {\EC} {\SN} problems,
such as directed {\DB} {\kOS} problems~\cite{Bansal2009},
matroid base and submodular flow problems~\cite{Kiraly2008},
and matroid intersection and optimization over lattice polyhedra~\cite{Bansal2010}.
All of them are based on ite\-rative rounding.
For applying iterative rounding to a problem with degree-bounds,
we need to show that
every basic optimal solution to an LP relaxation has a high
fractional variable or the subgraph induced by its support has a low
degree node on which a degree-bound is given.  Once this property is
proven, a bicriteria approximation algorithm can be obtained by
repeatedly rounding up a high fractional variable or dropping the
degree-bound on a low degree node.  See \cite{Lau2011} for a
survey on the iterative rounding method.

\subsubsection{{\ElC} and {\NC} {\SN} without degree-bounds}

Fleischer, Jain, and Williamson~\cite{Fleischer2006} showed that iterative rounding
achieves ratio $2$ for {\ElC} {\SN}, and also for {\NC} {\SN} with $k=2$.
Aazami, Cheriyan, and Laekhanukit~\cite{Aazami2010} presented an
instance of undirected {\kCS} (without degree-bounds) for
which the basic optimal solution to the standard LP relaxation has all
variables of value $O\left(\frac{1}{\sqrt{k}}\right)$.  
Their instance belongs to a special case called the {\em augmentation version}, in which the given
graph has a $(k-1)$-connected subgraph of cost zero.
On the other hand, several works showed that this augmentation version can be decomposed into
a small number $p$ of problems similar to {\kOS} each.
The bound on $p$ was subsequently improved \cite{Cheriyan2003,Fakcharoenphol2008,Kortsarz2005},
culminating in the currently best known bound $O\left(\log \frac{n}{n-k}\right)$ \cite{Nutov2009a},
that applies for both directed and undirected graphs.
When one applies this method for the general version, an additional factor
of $O(\log k)$ is invoked, giving the approximation ratio
$O\left(\log k \log \frac{n}{n-k}\right)$ \cite{Nutov2009a}.
Cheriyan and V\'egh \cite{CVegh2012} showed that, for undirected graphs with $n \geq k^3(k-1)+k$,
this $O(\log k)$ factor can be saved: After solving
only two {\kOS} instances,
iterative rounding gives a $2$-approximation by the work of \cite{Fleischer2006,Jackson2005}.
This gives ratio $6$ for undirected graphs with $n \geq k^3(k-1)+k$.

The decomposition approach was also used for other {\SN} problems.
In \cite{Nutov2009b} it is shown that the augmentation version of
{\sf Rooted} {\NC} {\SN} can be decomposed into $p=O(k)$ instances of
a problem that is similar to {\ElC} {\SN},
while in \cite{Chuzhoy2009} it is shown that
{\NC} {\SN} can be decomposed into $p=O(k^3 \log n)$ instances of {\ElC} {\SN}.
In \cite{Laekhanukit2011,Nutov2011}, it is shown that the augmentation version of
{\SkCS} with $k\leq (1-\epsilon)T$ and $0<\epsilon <1$
can be decomposed into $\frac{1}{\epsilon} O(\log k)$ instances of {\sf Rooted} {\SN}.

Summarizing, many {\SN} problems can be decomposed into {\kOS} and
{\ElC} {\SN} problems.
Thus algorithms for various {\SN} problems 
can be derived from those
for {\kOS} and
{\ElC} {\SN} problems.

\subsubsection{{\DB} {\ElC} and {\NC} {\SN}}

Despite the success of iterative rounding,
{\DB} {\SN} with node- and element-connectivity requirements
still remain difficult to address with this method.
Both \cite{Lau2009} and \cite{Lau2008} mention that their algorithms for edge-connectivity 
extend to element-connectivity, but they assumed that degree-bounds were given on terminals only.
In \cite{Lau2009} it is also shown that undirected {\DB} {\SkCS} with $k=\Omega(n)$ admits no
$2^{\log^{1-\epsilon}n}b(v)$ degree approximation
unless ${\rm NP} \subseteq {\rm DTIME}(n^{{\rm polylog}(n)})$.
For the {\DB} {\kCS} problem without costs,
Feder, Motwani, and Zhu~\cite{Feder2006} presented
an $O(k\log n \cdot b(v))$-approximation algorithm, which runs in $n^{O(k)}$ time.
Khandekar, Kortsarz, and Nutov \cite{Khandekar2011}
proposed a $(4,6b(v)+6)$-approximation algorithm for {\sf Degree-Bounded $2$-Connected Subgraph},
using iterative rounding.
Nutov~\cite{Nutov2012} extended the idea of \cite{Khandekar2011} to obtain ratio 
$(O(\log k), O(2^k) \cdot b(v))$ for {\DB} {\kOS} and {\DB} {\ElC} {\SN}.

As in the problems without degree-bounds,
many {\SN} problems with degree-bounds are decomposed into {\DB} {\kOS} and
{\ElC} {\SN} problems.
However, as indicated in \cite{Nutov2012}, compared to edge-connectivity problems,
there is a substantial difficulty in proving that
iterative rounding achieves a good result for these problems.
We resolve this difficulty by introducing several novel ideas.
Moreover, we believe that it is worthwhile investigating the iterative rounding approach for
node-connectivity requirements.
One reason is that
iterative rounding seems to be a promising approach for {\DB} {\SN} problems,
as we demonstrate in this paper.
A second reason is that it may give new insights for
improving the approximability of {\NC} {\SN} problems
(without degree-bounds)
with rooted requirements, subset $k$-connectivity requirements, and general requirements.

\subsection{Our results} \label{ss:results}

\subsubsection{New analysis of iterative rounding}

We show that iterative roun\-ding works well for {\DB} {\kOS} and 
{\DB} {\ElC} {\SN} problems.
Our main results for these two problems are summarized in the following two theorems.

\begin{theorem} \label{t:kOS}
{\DB} {\sf Directed} {\kOS} admits approximation ratio
$\left(\alpha,\alpha b(v)+\left\lceil\frac{2(k-1)}{\alpha-1}\right\rceil+1\right)$ for any integer $\alpha \geq 2$.
\end{theorem}

\begin{theorem} \label{t:ElC}
{\DB} {\ElC} {\SN} admits the following approximation ratios.
\begin{itemize}
\item[\em (i)]
$\left(\alpha,\alpha b(v)+\left\lceil 4\frac{k+1}{\alpha-2}\right\rceil +4\right)$ for any integer $\alpha \geq 4$.
\item[\em (ii)]
$(\infty,2b(v)+1.5k^2 +4.5k+9)$. 
\end{itemize}
\end{theorem}
Note that in Theorem~\ref{t:ElC}, the degree approximation in part (ii) may be better than the one 
in part (i) if $b(v)>k^2$.
The ratios in Theorems \ref{t:kOS} and \ref{t:ElC} improve the ratio $(O(\log k),O(2^k b(v)))$ of \cite{Nutov2012}.
In a preliminary version~\cite{Fukunaga2012} of this paper,
we gave a result similar to Theorem~\ref{t:kOS} with $\alpha=2$, but 
the additive term in the degree bounds is slightly improved in
Theorem~\ref{t:kOS}.
For {\DB} {\ElC} {\SN},
\cite{Fukunaga2012} also gave approximation ratios
$(4k-1,(4k-1)b(v)+O(k))$ and $(\infty,6b(v)+O(k^2))$.
Thus, part (i) of Theorem~\ref{t:ElC} improves the former ratio even if
$\alpha=4k-1$,
and part (ii) of Theorem~\ref{t:ElC} improves the coefficient of $b(v)$ in
the latter ratio.

All results in Theorems~\ref{t:kOS} and \ref{t:ElC} also bound the
integrality gaps of an LP relaxations in the same ratios.
Lau and Singh~\cite{Lau2008}
gave an example that indicates that the integrality gap of an LP
relaxation for the {\DB} {\EC} {\SN} is at least $(2, b(v)+\Omega(k))$.
In this example, all nodes are terminals, and hence the
element-connectivity is equivalent to the edge-connectivity. Thus the
integrality gap given by part (i) of Theorem~\ref{t:ElC} is tight up to a constant factor.

In \cite{Nutov2012} it is shown that ratio $(\alpha,\beta(b(v)))$ for
{\DB} {\sf Directed} {\kOS} implies ratio $(2\alpha,\beta(b(v))+k)$
for the undirected case. Thus Theorem~\ref{t:kOS} implies for the undirected case
the ratio $(2\alpha,\alpha b(v)+O(k))$ for any integer $\alpha \geq 2$.
In particular, for $\alpha=2$ the ratio is $(4,2b(v)+O(k))$.

In addition,
using known decompositions,
Theorems \ref{t:kOS} and \ref{t:ElC}
present the following results for several undirected variants of {\SN} problems.

\begin{theorem} \label{t:reductions}
{\SN} problems on undirected graphs admit the following approximation ratios for any integer $\alpha \geq 1$.
\begin{itemize}
\item[\rm (i)]
$O(k^3\log |T|) \cdot (\alpha, \alpha b(v)+k/\alpha)$ for {\DB} {\NC} {\SN}.
\item[\rm (ii)]
$O(k\log k) \cdot (\alpha, \alpha b(v)+k/\alpha)$ for {\DB} {\sf Rooted} {\SN}.
\item[\rm (iii)]
$\frac{1}{\epsilon}O(k \log^2 k) \cdot (\alpha, \alpha b(v)+k/\alpha)$
for {\DB} {\SkCS} with $k \leq (1-\epsilon)|T|$ and $0<\epsilon <1$.
\end{itemize}
\end{theorem}

\subsubsection{Improving reduction of {\DB} {\kCS} problem}
Next, we consider the {\DB} {\kCS} problem.
In \cite{Nutov2012}, it is shown that if {\DB} {\kOS} admits
approximation ratio $(\alpha,\beta(b(v)))$, then {\DB} {\kCS} admits
approximation ratio $(\alpha+O(k),\beta(b(v))+O(k^2))$.
We improve this reduction as follows.

\begin{theorem} \label{t:new-reduction}
If {\DB} {\kOS} admits approximation ratio $(\alpha,\beta(b(v)))$,
then {\DB} {\kCS} admits approximation ratio $(\mu \alpha+O(k),\beta(b(v))+O(k\sqrt{k}))$,
where $\mu=1$ for undirected graphs and $\mu=2$ for digraphs.
Consequently, {\DB} {\kCS} admits approximation ratio
$(2\mu +O(k), 2b(v)+O(k\sqrt{k}))$.
\end{theorem}

\subsubsection{Improving the constant approximation algorithm for {\kCS}}

Cheriyan and V\'egh \cite{CVegh2012} showed that {\kCS}
admits a 6-approximation algorithm if the given graph is undirected and
$n\geq k^3(k-1)+k$.
Their algorithm uses an algorithm for {\sf Undirected} {\kOS} as a subroutine.

It is not difficult to see that their results can be extended to the degree-bounded setting.
Our $(4,2b(v)+O(k))$-approximation algorithm for
{\DB} {\sf Undirected} {\kOS} presents a
$(12,8b(v)+O(k))$-approximation algorithm 
for {\DB} {\sf Undirected} {\kCS} under the same assumption as \cite{CVegh2012}.
We will explain this in Section~\ref{s:CV}.

In addition, we prove that
the assumption $n \geq k^3(k-1)+k$ in \cite{CVegh2012} can be weakened;
the bound $k^3(k-1)+k=O(k^4)$ is 
improved to $k(k-1)(k-1.5)+k=O(k^3)$
if we have no degree-bounds,
and to $2k(k-1)(k-0.5)+k$ in the degree-bounded setting.
Thus we obtain the following results.

\begin{theorem} \label{t:kCS}
For undirected graphs with $n \geq k(k-1)(k-1.5)+k$,
{\kCS} admits a $6$-approximation algorithm, and if $n \geq 2k(k-1)(k-0.5)+k$ then
{\DB} {\kCS} admits approximation ratio $(12,8b(v)+O(k))$.
\end{theorem}

\subsubsection*{Remark}
When this paper was in submission,
Ene and Vakilian~\cite{EneV2014} published a paper that presents several
improved results for {\DB} {\SN}.
They gave an $(3,6b(v)+5)$-approximation for {\DB} {\EC} {\SN},
$(3,6b(v)+3)$-approximation for {\DB} {\kOS},
and $(15,34b(v)+17)$ for {\DB} {\kCS} with $|V|=\Omega(k^3)$.

\subsection{Organization}
The rest of this paper is organized as follows.
In Section~\ref{s:bisets} we formulate Theorems \ref{t:kOS} and \ref{t:ElC} in terms of biset functions,
see Theorems \ref{t:kOS'} and \ref{t:ElC'}, respectively.
In Section~\ref{s:rounding} we describe the iterative rounding algorithm that we use, and 
formulate the latter two theorems in terms of extreme points of appropriate polytopes;
see Theorems \ref{t:kOS''} and \ref{t:ElC''}, respectively.
These two theorems are the key ingredients in proving Theorems \ref{t:kOS} and \ref{t:ElC}, respectively;
we prepare some tools for proving them in Section~\ref{s:extr-points} and prove them formally 
in Section~\ref{s:high-value}.
In Section~\ref{s:reductions}, we show that Theorem~\ref{t:reductions}
follows from Theorems~\ref{t:kOS} and \ref{t:ElC}.
In the subsequent two 
Sections~\ref{s:new-reduction} and \ref{s:kCS},
we prove
Theorems \ref{t:new-reduction} and \ref{t:kCS},  respectively.

\section{Biset edge-covering formulation of {\SN} problems (Theorems \ref{t:kOS} and \ref{t:ElC})} \label{s:bisets}

\subsection{Biset LP relaxation}
We use a standard setpair LP relaxation, due to Frank and Jord\'an~\cite{Frank1995},
but we formulate it in equivalent but more convenient
terms of {\em bisets}, as was suggested by Frank \cite{Frank2009},
and used in several other recent papers \cite{Nutov2009a,Nutov2009b,Nutov2011,Nutov2012}.

A {\em biset} is an ordered pair $\hat{S}=(S, S^+)$ of subsets
of $V$ such that $S\subseteq S^+$.
$S$ is called the {\em inner-part} of $\hat{S}$ and
$S^+$ is called the {\em outer-part} of $\hat{S}$.
We call $S^+\setminus S$ the {\em boundary} of $\hat{S}$, denoted by $\Gamma(\hat{S})$.
In the case of undirected graphs, an edge $e$ covers a biset $\hat{S}$
if it has one end-node in $S$ and the other in $V \setminus S^+$,
and we denote by $\delta_E(\hat{S})$ the set of edges in $E$ covering $\hat{S}$.
In the case of directed graphs, an edge $e$ covers $\hat{S}$ if it enters $\hat{S}$,
namely, if $e$ has tail in $V \setminus S^+$ and head in $S$;
we denote by $\delta_E^-(\hat{S})$ the set of edges in $E$ that cover $\hat{S}$.
We also denote by $\delta_E^+(\hat{S})$ the set of edges in $E$ that leave $\hat{S}$,
namely, edges in $E$ with tail in $S$ and head in $V \setminus S^+$.

Let $\Vset$ denote the set of all bisets of a groundset $V$.
A graph $(V,F)$ satisfies the connectivity requirements if
$|\delta_F(\hat{S})| \geq f(\hat{S})$  for each $\hat{S} \in {\cal V}$
in undirected graphs, and
$|\delta^-_F(\hat{S})| \geq f(\hat{S})$ for each $\hat{S} \in {\cal V}$
in directed graphs, where $f$ is the biset function derived from the connectivity requirements;
in this case we say that the graph $(V,F)$ is {\em $f$-connected}.

Every set $S$ can be considered as the biset $(S,S)$.
Hence degree constraints can be represented by using bisets.
For a node $v$, define a biset $\hat{S}_v=(\{v\},\{v\})$.
Then the degree constraint on a node $v$ in undirected graphs is represented by
$|\delta_F(\hat{S}_v)| \leq b(v)$.
In digraphs,
the out-degree constraint on $v$ is $|\delta^+_F(\hat{S}_v)| \leq b(v)$, and
the in-degree  constraint on $v$ is $|\delta^-_F(\hat{S}_v)| \leq b^-(v)$.
We sometimes abuse the notation to identify a node $v \in V$ as the biset $(\{v\}, \{v\})$.

We consider the following generic problem.

\vspace*{0.2cm}

\begin{center}
\fbox{
\begin{minipage}{0.960\textwidth}
\noindent
{\bf {\em Degree Bounded $f$-Connected Subgraph}} \\
A graph $G=(V,E)$ with edge/arc costs $c\colon  E \rightarrow \Rset_+$,
a biset function $f$ on ${\cal V}$,
and degree-bounds $b\colon B \rightarrow \Zset_+$ on a subset $B$ of $V$ are given.
The goal is to find a minimum cost edge set $F \subseteq E$ such that $(V,F)$
is $f$-connected, and the degree/out-degree of each $v \in B$ is at most $b(v)$.
\end{minipage}
}
\end{center}

\medskip
Let $x(e) \in [0,1]$ be a variable indicating whether an edge $e \in E$
is chosen to the solution or not.
Given an edge-set $F$ let $x(F)=\sum_{e \in F} x(e)$.
Our LP relaxation for the degree-bounded $f$-connected subgraph problem is 
$\min\{c \cdot x\colon x \in P(f,b,E)\}$, where $P(f,b,E)$ is a polytope defined as follows.
In the case of directed graphs $P(f,b,E)$ is defined by the constraints
\begin{equation*} 
\begin{array}{lll}
& x(\delta^-_{E}(\hat{S})) \geq f(\hat{S}) \ \ & \mbox{\  for each $\hat{S} \in \Vset$,}  \\
& x(\delta^+_{E}(v)) \leq b(v)                 & \mbox{\  for each $v \in B$,}            \\
& 0 \leq x(e) \leq 1                           & \mbox{\  for each $e \in E$.}            
\end{array}
\end{equation*}
In the case of undirected graphs, $P(f,b,E)$ is defined 
by replacing both $\delta^-_{E}(\hat{S})$ and $\delta^+_{E}(\hat{S})$ with $\delta_E(\hat{S})$.

Given a biset function $f$ we denote
$\gamma=\gamma_f=\max_{f(\hat{S})>0}|\Gamma(\hat{S})|$.
Given a biset function $f$ and an edge-set $J$ the {\em residual biset function} of $f$ is
$f_J(\hat{S})=f(\hat{S})-|\delta^-_J(\hat{S})|$ in the case of directed graphs and 
$f_J(\hat{S})=f(\hat{S})-|\delta_J(\hat{S})|$ in the case of undirected graphs.
Given a parameter $\alpha \geq 1$ the {\em residual degree bounds} are 
$b_J^\alpha(v)=b(v)-|\delta^+_J(v)|/\alpha$ in the case of directed graphs and
$b_J^\alpha(v)=b(v)-|\delta_J(v)|/\alpha$ in the case of undirected graphs.

\subsection{Intersecting supermodularity and {\sf Directed} {\kOS}}

Let $\hat{X}=(X,X^+)$ and $\hat{Y}=(Y,Y^+)$ be two bisets.
Define
\[
\hat{X} \cap \hat{Y}= (X\cap Y, X^+ \cap Y^+) \text{\ and\ } \hat{X} \cup \hat{Y}=(X\cup Y, X^+ \cup Y^+).
\]

For a biset function $f$ and bisets $\hat{X}$ and $\hat{Y}$,
the {\em supermodular inequality} is defined as 
\begin{equation}
\label{e:super}
f(\hat{X})+f(\hat{Y}) \leq f(\hat{X} \cap \hat{Y}) + f(\hat{X} \cup \hat{Y}).
\end{equation}
A biset function $f\colon  \Vset \rightarrow \Zset$ is {\em intersecting supermodular}
if any $\hat{X},\hat{Y}$ with $X \cap Y \neq \emptyset$
satisfy the supermodular inequality \eqref{e:super}.

We assume that for any $G=(V,E)$, $J$, and $\alpha \geq 1$, one can find an extreme point solution to
$\min\{c \cdot x\colon  x \in P(f_J,b_J^\alpha,E)\}$.
Under this assumption,
we prove the following theorem.

 \begin{theorem} [Implies Theorem~\ref{t:kOS}] \label{t:kOS'}
If $f$ is intersecting supermodular then directed {\DB} {\fCS} admits approximation ratio
$(\alpha,\alpha b(v)+\left\lceil\frac{2\gamma}{\alpha-1}\right\rceil+1)$
  for any integer $\alpha \geq 2$,
  where $\gamma=\gamma_f=\max_{f(\hat{S})>0}|\Gamma(\hat{S})|$.
\end{theorem}

This theorem implies Theorem~\ref{t:kOS}.
For $s \in V$ and an integer $k \geq 1$, define a biset function $g\colon \Vset \rightarrow \Zset$ as
\begin{equation*} 
 g(\hat{S}) =
 \begin{cases}\displaystyle
  k - |\Gamma(\hat{S})| & \mbox{if $S \neq \emptyset$ and $s \not\in S^+$,} \\
  0                     & \mbox{otherwise.}
 \end{cases}
\end{equation*}
Then a digraph $(V,F)$ is $k$-outconnected from $s$ 
(namely, satisfies $\kappa(s,v)\geq k$ for each $v \in V \setminus \{s\}$) if and only if
$|\delta^-_F(\hat{S})| \geq g(\hat{S})$ for each $\hat{S} \in \Vset$,
namely if and only if $(V,F)$ is $g$-connected.
Thus $g$ represents the $k$-out-connectivity requirements.

It is known that $g$ is intersecting supermodular~\cite{Frank2009}.
Observe that $\gamma_g \leq k-1$. Moreover, 
an extreme point solution to
$\min\{c \cdot x\colon  x \in P(g_J,b_J^\alpha,E)\}$ can be computed in
polynomial time
(we omit the somewhat standard implementation details).
Therefore,
Theorem~\ref{t:kOS} is obtained by applying Theorem~\ref{t:kOS'}
to $g$.

\subsection{Skew supermodularity and {\ElC} {\SN}}
For two bisets $\hat{X}=(X,X^+)$ and $\hat{Y}=(Y,Y^+)$,
define
\[
 \hat{X} \setminus \hat{Y} = (X\setminus Y^+, X^+ \setminus Y).
\]
For a biset function $f$ and bisets $\hat{X}$ and $\hat{Y}$,
the {\em negamodular inequality} is defined as 
\begin{equation}
 \label{e:posi}
 f(\hat{X})+f(\hat{Y})  \leq  f(\hat{X} \setminus \hat{Y}) + f(\hat{Y} \setminus \hat{X}).
\end{equation}
 $f$ is called (\emph{positively}) {\em skew supermodular} if 
 the supermodular inequality \eqref{e:super} or the negamodular
 inequality \eqref{e:posi}
 holds for any bisets $\hat{X},\hat{Y}$ with $f(\hat{X})>0$ and $f(\hat{Y})>0$.

We prove the following theorem under the same assumption as Theorem~\ref{t:kOS'}.

\begin{theorem} [Implies Theorem~\ref{t:ElC}] \label{t:ElC'}
If $f$ is skew supermodular then undirected {\DB} {\fCS} admits the following approximation ratios.
\begin{itemize}
\item[{\em (i)}]
$\left(\alpha, \alpha b(v)+\left\lceil 4\frac{\gamma+2}{\alpha-2}\right\rceil +4\right)$ 
for any integer $\alpha \geq 4$. 
\item[{\em (ii)}]
	     $(\infty,2b(v)+1.5\gamma^2 +7.5\gamma+15)$.
\end{itemize}
 Here, $\gamma=\gamma_f=\max_{f(\hat{S})>0}|\Gamma(\hat{S})|$.
\end{theorem}

\medskip
For $T \subseteq V$ and $r\colon T \times T \rightarrow \Zset_+$
define a biset function $h\colon  \Vset \rightarrow \Zset$
as
\begin{equation*} 
 h(\hat{S}) =
 \begin{cases}\displaystyle
  \max_{u \in S \cap T, v \in T\setminus S^+} r(u,v) -|\Gamma(\hat{S})| &
  \mbox{if $S \cap  T \neq \emptyset \neq T \setminus S^+$ and $T \cap \Gamma(\hat{S})=\emptyset$,} \\
  0                                                                     & \mbox{otherwise.}
 \end{cases}
\end{equation*}
By a ``mixed-connectivity'' version of Menger's Theorem (see, e.g. \cite{Kortsarz2007}),
an undirected graph $(V,F)$ satisfies
$\lambda^T(u,v)\geq r(u,v)$ for each $u,v \in T$ if and only if
$|\delta_F(\hat{S})| \geq h(\hat{S})$ for each $\hat{S} \in \Vset$,
namely, if and only if $(V,F)$ is $h$-connected.
Thus $h$ represents element-connectivity requirements.

By \cite{Fleischer2006}, $h$ is skew supermodular.
Moreover, 
$\gamma_h \leq k-1$ holds, and
an extreme point solution to
$\min\{c \cdot x\colon  x \in P(g_J,b_J^\alpha,E)\}$ can be computed in
polynomial time.
Therefore,
Theorem~\ref{t:ElC} is obtained by applying Theorem~\ref{t:ElC'} to $h$.

\section{Iterative rounding algorithms 
(Theorems \ref{t:kOS'} and \ref{t:ElC'})} \label{s:rounding}

Here we describe the version of the iterative rounding method we use, 
and formulate Theorems \ref{t:kOS'} and \ref{t:ElC'} in terms of extreme points of appropriate polytopes.
To apply iterative rounding, we define $J \subseteq E$ as
the set of edges that have already been chosen as a part of the current
solution by the algorithm. We also denote by $\deg_E(v)$ the out-degree/degree of $v$ w.r.t. $E$, so
$\deg_E(v)=|\delta^+_E(v)|$ in the case of digraphs and
$\deg_E(v)=|\delta_E(v)|$ in the case of undirected graphs.
The algorithms have three parameters $\alpha$, $\beta$, and $\sigma$.
Parameter $\sigma$ must satisfy $\sigma \leq \alpha$, and is usually set
to $\alpha$ or $0$.

\begin{center}
\fbox{
\begin{minipage}{0.960\textwidth}
\begin{description}\setlength{\itemsep}{0pt}
 \item[Algorithm {\sc IteRounding}]
 \item[Input:] A graph $G=(V,E)$, $B\subseteq V$,
	    degree-bounds $b\colon B \rightarrow \Zset_+$, edge costs
	    $c\colon E\rightarrow \Zset_+$, a biset function $f\colon \Vset \rightarrow \Zset$, 
	    and integers $\alpha \geq 1$, $\beta \geq 0$, and $\sigma \leq \alpha$.
 \item[Output:] An $f$-connected subgraph $(V,J)$ of $G$.
 \item[Step 1:] $J := \emptyset$.
 \item[Step 2:] Compute a basic optimal solution  $x^*$ to $\min\{c \cdot x\colon  x \in P(f_J,b_J^\alpha,E)\}$.
 \item[Step 3:] If there is $e \in E$ such that $x^*(e)=0$             then remove $e$ from $E$.
 \item[Step 4:] If there is $e \in E$ such that $x^*(e) \geq 1/\alpha$ then move $e$ from $E$ to $J$.
 \item[Step 5:] If there is $v \in B$ such that $\deg_E(v) \leq \sigma b_J^\alpha(v)+\beta$
	              then remove $v$ from $B$.
 \item[Step 6:] If $E \neq \emptyset$, return to Step~2.
	              Otherwise, output $(V,J)$.
\end{description}
\end{minipage}
}
\end{center}

\vspace*{0.2cm}

The performance of various versions of this algorithm are analyzed in several papers, c.f. 
\cite{Bansal2009,Lau2009,Lau2011,Lau2008,Nutov2012}.
Assume that at each iteration, there exists 
an edge $e \in E$ as in Steps 3 or 4, or 
a node $v \in B$ as in Step 5. 
Then the algorithm {\sc IteRounding} computes an edge set $J$ of cost 
at most $\alpha$ times the optimal,
such that $\deg_J(v) \leq \alpha b(v)+\beta$ for every $v \in B$
 and 
if $\sigma=0$ then $\deg_J(v) \leq \alpha b(v)+\max\{\beta-1,0\}$ for
every $v \in B$
(since $b(v) \geq 1$ for $v \in B$; for details see e.g.~\cite{Nutov2012}).

We say that $x$ is \emph{maximal} in a polyhedron if the polyhedron contains no
point $y$ such that $x(e) \leq y(e)$ for all $e \in E$ and $x(e') <
y(e')$ holds for some $e' \in E$. 
If we care only about the degree approximation, as in part (ii) of 
Theorems \ref{t:ElC} and \ref {t:ElC'}, then
we define $x^*$ computed in Step~2 as a basic optimal solution to
$\max\{\sum_{e\in E}x(e) \colon x \in  P(f_J,b_J^\alpha,E)\}$;
thus, $x^*$ is maximal in $P(f_J,b_J^\alpha,E)$ in this case.
Note that at Step~4 we can move from $E$ to $J$ 
any edge $e$ that has no tail/end-node in $B$, without changing the approximability of the degrees.

Note that an arbitrary set $F$ of directed edges satisfies
$|\delta^-_F(\hat{X})| + |\delta^-_F(\hat{Y})| \geq |\delta^-_F(\hat{X} \cap
\hat{Y})| + |\delta^-_F(\hat{X}\cup \hat{Y})|$
for any $\hat{X},\hat{Y} \in \Vset$.
Similarly, an arbitrary set $F$ of undirected edges satisfies both
$|\delta_F(\hat{X})| + |\delta_F(\hat{Y})| \geq
|\delta_F(\hat{X} \cap \hat{Y})| + |\delta_F(\hat{X}\cup \hat{Y})|$
and
$|\delta_F(\hat{X})| + |\delta_F(\hat{Y})| \geq
|\delta_F(\hat{X} \setminus \hat{Y})| + |\delta_F(\hat{Y}\setminus \hat{X})|$
for any $\hat{X},\hat{Y} \in \Vset$.
These facts can be verified by counting contributions of edges in both
sides (see, e.g., \cite{Lau2011}). 
Hence if $f$ is intersecting supermodular and the graph is directed, or
if $f$ is skew supermodular and the graph is undirected, 
so is its residual function $f_J$.

We will prove the following property of extreme points solutions.

\begin{theorem} [Implies Theorem~\ref{t:kOS'}] \label{t:kOS''}
Let $x^*$ be an extreme point of $P(f,b,E)$ where $G$ is a directed graph and 
$f$ is an intersecting supermodular biset function on $V$.
Then for any integer $\alpha \geq 2$, there is $e \in E$ with
$x^*(e)=0$ or $x^*(e) \geq 1/\alpha$, or there is
$v \in B$ with $|\delta^+_E(v)| \leq \alpha b(v)+\left\lceil
 \frac{2\gamma}{\alpha-1}\right\rceil+1$.
Here, $\gamma=\gamma_f=\max_{f(\hat{S})>0}|\Gamma(\hat{S})|$. 	     
\end{theorem}

\begin{theorem} [Implies Theorem~\ref{t:ElC'}] \label{t:ElC''}
Let $x^*$ be an extreme point of $P(f,b,E)$ where $G=(V,E)$ is an undirected graph and 
 $f$ is a skew supermodular biset function on $V$.
Recall that $\gamma=\gamma_f=\max_{f(\hat{S})>0}|\Gamma(\hat{S})|$. 	     
 Then the following holds.
\begin{itemize}
\item[\rm (i)]
For any integer $\alpha \geq 4$,
there is $e \in E$ with $x^*(e)=0$ or $x^*(e) \geq 1/\alpha$, or
there is $v \in B$ with $|\delta_E(v)| \leq \left\lceil 4\frac{\gamma+2}{\alpha-2}\right\rceil + 5$. 
\item[\rm (ii)]
	     If $x^*$ is maximal in $P(f,b,E)$ and
	     every edge in $E$ is incident to some node in $B$, then
	     there is $e \in E$ with $x^*(e)=0$ or $x^*(e) \geq 1/2$, or
	     there is $v \in B$ with $|\delta_E(v)| \leq
	     1.5\gamma^2+7.5\gamma+16$.
\end{itemize}
\end{theorem}

\subsection*{Proof ideas in Theorems~\ref{t:kOS''} and \ref{t:ElC''}}

Roughly speaking, Theorems \ref{t:kOS''} and \ref{t:ElC''} are based on two
new ideas. One is a structural result on laminar biset families.
In the edge-connectivity problems,
each extreme point solution of an LP relaxation is characterized by a 
laminar set family,
and this fact always plays a key role in the analysis of iterative
rounding algorithms.
For proving  Theorems \ref{t:kOS''} and \ref{t:ElC''},
we extend concept of laminarity to biset families.
Indeed, we have two definitions of laminar biset families, 
we call one of them simply a laminar family, and the other a
strongly laminar family.
We will show in Lemma~\ref{l:rank}
that an extreme point
of $P(f,b,E)$ is defined by a laminar biset family if $f$ is
intersecting supermodular and the graph is directed,
and by a strongly laminar biset family if $f$ is skew supermodular and
the graph is undirected.
Despite this
observation, we still have some difficulty in proving
Theorems~\ref{t:kOS''} and \ref{t:ElC''}. This is because
a node is possibly shared by more than one biset even in these biset families.
This is the main reason why 
Lau et al.~\cite{Lau2009} and Lau and Singh~\cite{Lau2008} needed to assume that the
degree-bounds are given on terminals only, and 
\cite{Nutov2012} needed an exponential approximation factor in the degree-bounds.
We overcome this
by presenting a bound on the number of bisets that share
degree-bounded nodes (Lemma~\ref{l:count}).

The other new idea is to use two token counting strategies depending upon the size of the
laminar biset family and the number of tight degree nodes.
A standard way to analyze iterative rounding is to compare the
number of variables and the number of constraints that define
an extreme point solution.
In our problem, 
each variable corresponds to an edge, and
an extreme point solution is defined by 
constraints corresponding to a
biset in a laminar or strongly laminar family, or to a 
tight degree node.
Hence we assign tokens to each edge,
and count the number of tokens in terms of the size of the biset family
and the number of tight degree nodes.
In Theorem~\ref{t:kOS''} and part (i) of Theorem~\ref{t:ElC''},
if the number of tight degree nodes is smaller compared with the number
of leaf bisets in the biset family, we can almost ignore the
tight degree nodes, and hence the token counting for the problem without
degree-bounds can be modified for such a case.
For the other case, we apply a token counting using the above structural
result on a laminar or strongly laminar biset family.
As for part (ii) of Theorem~\ref{t:ElC''}, we do the opposite.
If the number of tight degree nodes is smaller, we apply a known token distribution after careful
redistribution of tokens. 
If the number of tight degree nodes is larger, we use the fact that
each edge is incident to some node in $B$.

These ideas have been already given in a preliminary version
\cite{Fukunaga2012} of this paper.
By generalizing and simplifying them, the ratios were slightly
improved in Theorem~\ref{t:kOS''} and part (ii) of Theorem~\ref{t:ElC''}.
In \cite{Fukunaga2012},
we did not apply the idea of using two token strategies for
bicriteria approximation of {\DB} {\ElC} {\SN}.
This is why
\cite{Fukunaga2012} did not have a result corresponding to
part (i) of Theorem~\ref{t:ElC''} with a constant $\alpha$.

Proofs of Theorems \ref{t:kOS''} and \ref{t:ElC''} will be presented in the
subsequent two sections.
In Section~\ref{s:extr-points}, we discuss laminarity of biset families.
In Section~\ref{s:high-value}, we
prove the theorems by giving token counting arguments.

\section{Laminar biset families} \label{s:extr-points}

A set family ${\cal F}$ is {\em laminar} if for any $X,Y \in {\cal F}$,
either $X \subseteq Y$, $Y \subseteq X$, or $X \cap Y=\emptyset$.
Note that $X \subseteq Y$ or $Y \subseteq X$ holds if and only if
$\{X \cap Y,X \cup Y\}=\{X,Y\}$, and that $X \cap Y=\emptyset$ holds if and only if
$\{X \setminus Y,Y\setminus X\}=\{X,Y\}$.
Our aim is to establish that any extreme point of $P(f,b,E)$
with intersecting supermodular $f$ or with skew supermodular $f$ 
is defined by a  ``laminar'' family of bisets.
But how to define ``laminarity'' of bisets?

In the case of an intersecting supermodular $f$, it is natural to say that ${\cal F}$ is laminar if
$\{\hat{X} \cap \hat{Y},\hat{X} \cup \hat{Y}\}=\{\hat{X},\hat{Y}\}$
for any $\hat{X},\hat{Y} \in {\cal F}$ with $X \cap Y \neq \emptyset$.
In the case of a skew supermodular $f$, it is natural to say that ${\cal F}$ is laminar if
$\{\hat{X} \cap \hat{Y},\hat{X} \cup \hat{Y}\}=\{\hat{X},\hat{Y}\}$ or
$\{\hat{X} \setminus \hat{Y},\hat{Y} \setminus \hat{X}\}=\{\hat{X},\hat{Y}\}$ for any
$\hat{X},\hat{Y} \in {\cal F}$; we refer to this property as ``strong laminarity''.
Laminar biset families are used in \cite{Frank2009},
while strongly laminar biset families are defined in \cite{Fleischer2006} in terms of setpairs.
Following \cite{Nutov2012}, we formulate both these concepts in terms of bisets, 
by establishing an inclusion order (namely, a partial order) on bisets.

\begin{definition}
We say that a biset $\hat{Y}$ contains a biset $\hat{X}$ and
write $\hat{X} \subseteq \hat{Y}$ if $X \subseteq Y$ and $X^+ \subseteq Y^+$;
if also $\hat{X} \neq \hat{Y}$ then $\hat{X} \subset \hat{Y}$ and $\hat{Y}$ properly contains $\hat{X}$.
A biset family ${\cal F}$ is laminar (resp., strongly laminar)
if for any $\hat{X},\hat{Y} \in {\cal L}$ either
$\hat{X} \subseteq \hat{Y}$, 
$\hat{Y} \subseteq \hat{X}$, or
$X \cap Y=\emptyset$ (resp.,  $X \cap Y^+=\emptyset$ and $Y \cap X^+=\emptyset$).
\end{definition}

 \begin{figure}
  \centering
 \includegraphics[]{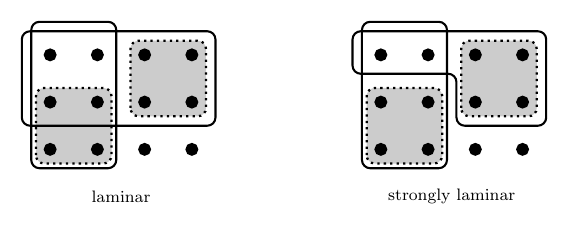}
  \caption{A laminar family and a strongly laminar family of bisets on twelve nodes.
  A filled region surrounded by a dotted line represents the inner-part of a biset, and a
  region surrounded by a solid line represents the outer-part.}
 \label{fig.laminar}
 \end{figure}

Figure~\ref{fig.laminar} illustrates examples of laminar and strongly
laminar biset families,
each of which consists of two bisets.
The family illustrated in the left is not strongly laminar because the outer-part
of a biset intersects the inner-part of the other.

Note that a strongly laminar biset family is laminar, and that the family of inner-parts
of a laminar biset family is a laminar set family.
The above inclusion order on bisets defines a forest structure of a laminar biset family.
In the rest of this paper, ``minimal'' and ``maximal'' are defined with respect to this inclusion order.
For a laminar biset family $\Fset$ and $\hat{X},\hat{Y} \in \Fset$, we say that $\hat{Y}$ is
the {\em parent} of $\hat{X}$ and $\hat{X}$ is a {\em child} of $\hat{Y}$
if $\hat{Y}$ is the minimal biset with $\hat{X} \subset \hat{Y}$.
A minimal biset in $\Fset$ is called a {\em leaf}.

Let $G=(V,E)$ be a graph, and let ${\cal F}$ be a biset family on $V$.
For $E' \subseteq E$ the characteristic vector of $E'$ is an $|E|$-dimensional vector whose component
corresponding to $e \in E$ is $1$ if $e\in E'$, and $0$ otherwise.
If $G$ is a directed graph, then let
$\chi^+_E({\cal F})$ denote the set of characteristic vectors of the edge sets in 
$\{\delta^+_E(\hat{S})\colon \hat{S} \in {\cal F}\}$ and 
$\chi^-_E({\cal F})$ is the set of characteristic vectors of 
the edge sets in $\{\delta^-_E(\hat{S})\colon \hat{S} \in {\cal F}\}$.
If $G$ is an undirected graph, let $\chi_E({\cal F})$ denote the set of characteristic vectors of 
the edge sets in $\{\delta_E(\hat{S})\colon \hat{S} \in {\cal F}\}$.
For $C \subseteq V$, $\chi_E(C)=\chi_E(\{(v,v)\colon v \in C\})$, and similarly $\chi^+_E(C)$ and $\chi^-_E(C)$ are defined. 
The following statement was proved for set-functions in \cite{Bansal2009,Lau2009}, and the proof for biset
functions is similar (e.g., see \cite{Nutov2012} for the case of a skew supermodular $f$). 

\begin{lemma} \label{l:rank}
Let $x$ be an extreme point of $P(f,b,E)$ with $0 < x_e < 1$ for all $e \in E \neq \emptyset$.
 Let $\mathcal{L}$ be a family of bisets, $C \subseteq B$, and suppose
 that they satisfy  the following conditions.
\begin{itemize}
\item[{\em (i)}]
If $G$ is a directed graph and $f$ is intersecting supermodular, then 
$x(\delta^-_E(\hat{S}))=f(\hat{S}) \geq 1$ for all $\hat{S} \in {\cal L}$,
$x(\delta^+_E(v))=b(v)$ for all $v \in C$, 
the vectors in $\chi^-_E({\cal L}) \cup \chi^+_E(C)$ are linearly independent,
and ${\cal L}$ is laminar.
\item[{\em (ii)}]
If $G$ is an undirected graph and $f$ is skew supermodular, then 
$x(\delta_E(\hat{S}))=f(\hat{S}) \geq 1$ for all $\hat{S} \in {\cal L}$,
$x(\delta_E(v))=b(v)$ for all $v \in C$, 
the vectors in $\chi_E({\cal L}) \cup \chi_E(C)$ are linearly independent,
and ${\cal L}$ is strongly laminar.
\end{itemize}
 If $\mathcal{L}$ and $C$ are inclusion-wise maximal,
 then $|\mathcal{L}|+|C|=|E|$.
\end{lemma}


The following parameter is defined in \cite{Nutov2012}.

\begin{definition}
Let ${\cal L}$ be a laminar biset family on $V$, let $\hat{S} \in {\cal L}$ and let $v \in V$.
We say that $\hat{S}$ {\em owns}   $v$ if $\hat{S}$ is the minimal biset in ${\cal L}$ with $v \in S$.
We say that $\hat{S}$ {\em shares} $v$ if $\hat{S}$ is a   minimal biset in ${\cal L}$ with $v \in \Gamma(\hat{S})$.
Let $\Delta_{\cal L}(v)$ denote the number of bisets in ${\cal L}$ that
share $v$.
\end{definition}

From the definition it follows that every node $v$ is owned by at most
one biset in a laminar family,
and that if two bisets in a laminar family share the same node $v$
then they are incomparable, namely, that none of them contains the other.
The definition also motivates the following lemma, which we need for
proving our structural result described in Lemma~\ref{l:count}.

\begin{lemma} \label{l:laminar}
Let ${\cal L}$ be a laminar biset family, and let $v \in V$.
Let ${\cal X}$ be a sub-family of ${\cal L}$ such that
$v \in \Gamma(\hat{X})$ for each $\hat{X} \in {\cal X}$, and the bisets in 
${\cal X}$ are pair-wise incomparable.
For each $\hat{X} \in {\cal X}$, let $\hat{X}'$ be a biset in ${\cal L}$ such that
$\hat{X} \subseteq \hat{X}'$ and $\hat{X'}$ contains no biset in ${\cal
 X} \setminus \{\hat{X}\}$
 (possibly $\hat{X}'=\hat{X}$). 
Then $v \in \Gamma(\hat{X}')$ or $v \in X'$
holds, and the latter holds for at most one biset in ${\cal X}'=\{\hat{X}' \colon
 \hat{X} \in {\cal X}\}$. Furthermore, if ${\cal L}$ is strongly laminar
 then the former holds for all bisets in  ${\cal X}'$.
\end{lemma}
\begin{proof}
 Since $\Gamma(\hat{X})$ is contained by the outer-part of $\hat{X}'$,
 $v \in \Gamma(\hat{X})$ is either in $\Gamma(\hat{X}')$ or in $X'$.
 Any two bisets $\hat{X}'$ and $\hat{Y}'$ in ${\cal X}'$ are
 incomparable, and hence $X' \cap Y' = \emptyset$ by the laminarity of
 $\Lset$.
 It follows from this fact that $v$ is contained by the inner-part of at
 most one biset in ${\cal X}'$.
 When $\Lset$ is strongly laminar, $X' \cap
 \Gamma(\hat{Y}')=\Gamma(\hat{X}') \cap Y'=\emptyset$
 holds for any $\hat{X}',\hat{Y}' \in {\cal X}'$, and hence 
 $v \in \Gamma(X')$ holds for all $\hat{X}' \in {\cal X}'$.
\end{proof}

 The following key statement will be used in our token countings.

\begin{lemma}[Structural Lemma] \label{l:count}
Let ${\cal L}$ be a  biset family on $V$, let $C \subseteq V$, let ${\cal E}$ be the set of leaves of ${\cal L}$,
and let
$\gamma=\max_{\hat{S} \in {\cal L}}|\Gamma(\hat{S})|$. Then
\begin{eqnarray}
\sum_{v \in C} \max\{\Delta_{\cal L}(v),1\}   & \leq & 2\gamma(|{\cal E}|-1)+|C| \ \
\mbox{ if } {\cal L} \mbox{ is laminar,}
\label{e:delta'} \\
\sum_{v \in C} \max\{\Delta_{\cal L}(v),1\} & \leq & \gamma|{\cal E}|+|C|   \ \ \ \ \ \ \ \ \ \
\mbox{ if } {\cal L} \mbox{ is strongly laminar.} \label{e:delta}
\end{eqnarray}
\end{lemma}
\begin{proof}
We may assume that $|\Delta_{\cal L}(v)| \geq 2$ for every $v \in C$,
as if $\Delta_{\cal L}(v) \leq 1$ for some $v \in C$, then
excluding $v$ from $C$ decreases both sides of each of (\ref{e:delta'}) and (\ref{e:delta}) by exactly one.
 We also assume that exactly one biset is maximal in $\cal L$;
 otherwise, we add the biset $(V,V)$ to $\cal L$, which has no effect
 on the claim.

 We prove (\ref{e:delta'}).
 Let ${\cal L}'$ be the family of bisets in ${\cal L}$
 whose parent has at least two children.
It is known and easy to prove by induction that $|{\cal L}'| \leq 2|{\cal E}|-2$.
 Let $v \in C$.
 ${\cal L}$ includes
 $\Delta_{\cal L}(v)$ incomparable bisets that contain $v$ in their
 boundaries.
 Let $\cal X$ denote a family of such bisets.
 For each biset $\hat{X} \in \cal X$,
 let $\hat{X}'$ be the minimal biset in $\cal L'$
 with $\hat{X} \subseteq \hat{X}'$.
 Note that $\hat{X'}$ includes no biset in ${\cal X} \setminus \{\hat{X}\}$.
 Lemma~\ref{l:laminar} indicates that 
 at most one biset in $\{\hat{X}' \colon \hat{X} \in \cal X\}$ includes $v$ in its inner-part.
 Thus $v$ belongs to the boundaries of at least $\Delta_{\cal L}(v)-1$ bisets in ${\cal L}'$.
 This implies
 \[
 \sum_{v \in C} \max\{\Delta_{\cal L}(v),1\} -|C| = \sum_{v \in C} (\max\{\Delta_{\cal L}(v),1\}-1) \leq
 \gamma |{\cal L}'| \leq 2\gamma(|{\cal E}|-1) \ .
 \]

We prove (\ref{e:delta}) by induction on $|{\cal E}|$.
Assume therefore that $|{\cal E}| \geq 2$, as otherwise $|\Delta_{\cal L}(v)| \leq 1$ for every $v \in C$.
Then there exists $\hat{S} \in {\cal L}$ such that $\hat{S}$ has at least two children,
but every proper descendant of $\hat{S}$ has at most one child. Let $\hat{R}$ be a child of $\hat{S}$,
 let $\hat{Z} \subseteq \hat{R}$ be a leaf of ${\cal L}$
(possibly $\hat{R}=\hat{Z}$), and let
${\cal P}=\{\hat{Y} \in {\cal L}\colon  \hat{Z} \subseteq \hat{Y} \subseteq \hat{R}\}$
be the ``chain'' from the child $\hat{R}$ of $\hat{S}$ to $\hat{Z}$. 
Since we assume that $|\Delta_{\cal L}(v)| \geq 2$ for every $v \in C$,
then by Lemma~\ref{l:laminar}, every node that is shared by some biset in ${\cal P}$
belongs to $\Gamma(\hat{R})$, so there are at most $|\Gamma(\hat{R})| \leq \gamma$ such nodes.
Hence excluding the bisets in ${\cal P}$ from ${\cal L}$ decreases the left hand side
of (\ref{e:delta}) by at most $\gamma$, while $|{\cal E}|$ decreases by $1$,
hence the right hand side of (\ref{e:delta}) decreases by $\gamma$.
\end{proof}

\section{Proof of Theorems \ref{t:kOS''} and \ref{t:ElC''}} \label{s:high-value}

Let $x=x^*$ be an extreme point solution to the corresponding biset LP relaxation, and
let ${\cal L}$ and $C$ be as in Lemma \ref{l:rank}.
Let ${\cal E}$ be the set of leaf bisets in ${\cal L}$.
For a biset $\hat{S} \in {\cal L}$, denote by ${\cal C}_S$ the set of children of $\hat{S}$, by
$E^+_S$ the set of edges in $E$ that cover $\hat{S}$ but not a child of $\hat{S}$, and by
$E^-_S$ the set of edges in $E$ that cover a child of $\hat{S}$ but not
$\hat{S}$;
see examples in Fig.~\ref{fig.edge}.
If $\hat{S} \in {\cal E}$, then $E^+_S=\delta^-_E(\hat{S})$ when $E$ is the
set of arcs, and $E^+_S=\delta_E(\hat{S})$ when $E$ is the
set of undirected edges.
$E^-_S= \emptyset$ if $\hat{S} \in {\cal E}$.
If $\hat{S}\in {\cal L} \setminus {\cal E}$, then
$E^+_S \cup E^-_S \neq \emptyset$ since otherwise
the vectors in $\chi_E(\{\hat{S}\}\cup {\cal C}_S)$ are linearly dependent.

 \begin{figure}
  \centering
 \includegraphics[scale=.65]{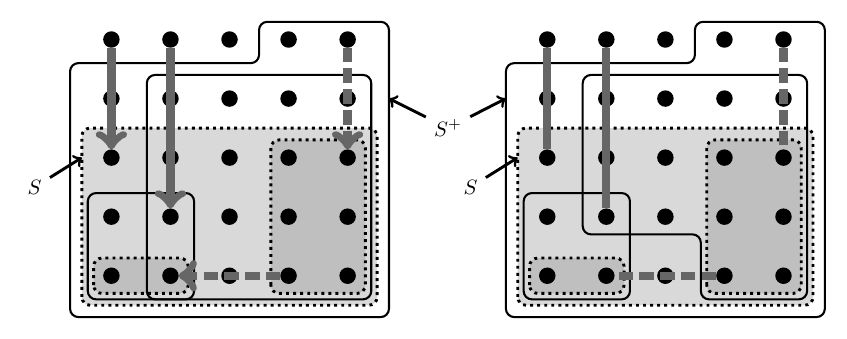}
  \caption{Examples of edges in $E^+_S$ and $E^-_S$.
  In the left example,
  edges are directed, and bisets are $\hat{S}$ and its two children that
  form a laminar family. In the right,
  edges are undirected and
  bisets are $\hat{S}$ and its two children that
  form a strongly laminar family. 
  Edges in $E_S^+$ and in $E_S^-$ are represented by gray solid lines
  and gray dotted lines, respectively.}
 \label{fig.edge}
 \end{figure}

\subsection{Intersecting supermodular $f$ and directed graphs (Theorem~\ref{t:kOS''})}

Assume for the sake of contradiction that $0 <  x(e)<1/\alpha$ for every $e \in E$.
Assign one token to every edge $e =uv \in E$,
putting $1-\alpha x(e) >0$ ``tail-tokens'' at $u$ and $\alpha x(e)>0$ ``head-tokens'' at $v$.
We will show that these tokens can be distributed such that every member of ${\cal L} \cup C$
gets one token, and some tokens are not assigned.
This gives the contradiction $|E|>|{\cal L}|+|C|$.

For every $\hat{S} \in {\cal L}$,
the amount of head-tokens of edges in $E^+_S$ and tail-tokens of edges in $E^-_S$ is
$\alpha x(E^+_S)+|E^-_S|-\alpha x(E^-_S)$.
Note that this is an integer, since $\alpha$ is an integer, and
since $x(E^+_S)-x(E^-_S)=f(\hat{S})-\sum_{\hat{R} \in {\cal
C}_S}f(\hat{R})$ follows from
$\delta_E^-(\hat{S})\setminus (\bigcup_{\hat{R} \in {\cal
C}_S}\delta_E^-(\hat{R}))=E^+_S$ and
$(\bigcup_{\hat{R} \in {\cal C}_S}\delta_E^-(\hat{R})) \setminus \delta_E^-(\hat{S}) =E^-_S$.
It is a positive integer since 
$x(E^+_S)>0$ if $E^+_S\neq \emptyset$, and 
$|E^-_S|-\alpha x(E^-_S) > 0$ if $E^-_S \neq \emptyset$.
Thus if we assign to every $\hat{S} \in {\cal L}$
the head-tokens of edges in $E^+_S$ and
the tail-tokens of edges in $E^-_S$,
then every member of ${\cal L}$ will get at least one token, and the tail-tokens entering the maximal
members of ${\cal L}$ are not assigned.
An edge belongs to $E_S^+$ for at most one biset $\hat{S} \in
\mathcal{L}$.
Thus each head-token is not counted twice.
The same goes for tail-tokens.

We note that we did not use
the assumption $\alpha \geq 2$ above, and hence this token
distribution is possible even if $\alpha=1$.
Hence, if $C=\emptyset$, and in particular if there are no degree bounds,
then this implies that the extreme points
of the polytope $P(f,b,E)$ are all integral.

\begin{lemma} \label{l:alpha'}
If $(\alpha-1)|{\cal E}| \geq |C|$ then there is $e \in E$ with $x(e) \geq 1/\alpha$.
\end{lemma}
 \begin{proof}
Every leaf $\hat{S}$ gets $\alpha x(\delta_E^-(\hat{S}))=\alpha f(\hat{S}) \geq \alpha$
head-tokens from edges in $E^+_S$. Hence we have $\alpha |{\cal E}|$ tokens at leaves.
By the assumption $(\alpha-1)|{\cal E}| \geq |C|$ we have
$|{\cal E}| +|C| \leq \alpha |{\cal E}|$,
hence the $\alpha |{\cal E}|$ tokens at leaves suffice to give one token to each member of ${\cal E} \cup C$.
Every non-leaf biset $\hat{S} \in {\cal L}$ gets
the head-tokens from edges in $E^+_S$ and
the tail-tokens from edges in $E^-_S$, so at least one token.
Consequently, every member of ${\cal L} \cup C$
gets one token, and the tail-tokens of the edges entering
the maximal members of ${\cal L}$ are not assigned.
This gives the contradiction $|E|>|{\cal L}|+|C|$.
 \end{proof}

\begin{lemma} \label{l:beta'}
If $|C| > (\alpha-1) |{\cal E}|$, then
there is $e \in E$ with $x(e) \geq 1/\alpha$ or
there is $v \in C$ with $|\delta^+_E(v)| \leq \alpha b(v)+\beta$, 
where $\beta = \left\lceil\frac{2\gamma}{\alpha-1}\right\rceil+1$.
\end{lemma}
\begin{proof}
Assume that
$|\delta^+_E(v)| \geq \alpha b(v)+\beta+1$ for every $v \in C$.
Then the amount of tail-tokens at each $v \in C$ is at least
$\alpha b(v)+\beta+1 - \alpha x(\delta^+_E(v))= \beta+1$.
Hence we have at least $\alpha |{\cal E}|+(\beta+1)|C|$ tokens at leaves and nodes in $C$.
From these tokens, we give one token to every leaf and $\Delta_{\cal L}(v)+2$ tokens
to every $v \in C$, and spare tokens remain.
This is possible,
since by \eqref{e:delta'} of Lemma~\ref{l:count} and the assumption $|C| > (\alpha-1) |{\cal E}|$,
\begin{eqnarray*}
|{\cal E}|+\sum_{v \in C} (\Delta_{\cal L}(v)+2)  &\leq&
|{\cal E}| +3|C| + 2\gamma(|{\cal E}|-1) \\
&=&\alpha |{\cal E}| + (2\gamma+1-\alpha)|{\cal E}| + 3|C|-2\gamma\\
&=&\alpha |{\cal E}| + |C| \left((2\gamma+1-\alpha)\frac{|{\cal E}|}{|C|} + 3\right)-2\gamma\\
&<&\alpha |{\cal E}| + |C| \left(\frac{2\gamma}{\alpha-1} + 2\right)-2\gamma\\
&\leq&\alpha |{\cal E}| + |C| (\beta+1) \ .
\end{eqnarray*}

Every $v \in C$ will keep one token, and from the remaining at least $\Delta_{\cal L}(v)+1$ tokens
$v$ will give one token to every biset that owns or shares $v$.
Now let $\hat{S} \in {\cal L}$ be a biset that is not a leaf
and that
does not own or share any node in $C$.
Then $\hat{S}$ gets
the head-tokens from edges in $E^+_S$ and
the tail tokens from edges in $E^-_S$, so at least one token as argued above.
Consequently, every members of ${\cal L} \cup C$
gets one token. This gives the contradiction $|E|>|{\cal L}|+|C|$.
\end{proof}

Theorem~\ref{t:kOS''} follows by combining Lemmas \ref{l:alpha'} and \ref{l:beta'}.

\subsection{Skew supermodular $f$ and undirected graphs (Part (i) of Theorem~\ref{t:ElC''})}

We now consider Theorem~\ref{t:ElC''}.
Recall that edges are undirected and $\mathcal{L}$ is 
strongly laminar in this theorem.
We deduce part (i) of Theorem~\ref{t:ElC''} from the following two lemmas.

\begin{lemma} \label{l:theta}
If $(\theta-1)|{\cal E}| \geq |C|$ for an integer $\theta \geq 2$,
then there is $e \in E$ with $x(e) \geq \frac{1}{2\theta}$.
\end{lemma}
\begin{proof}
We generalize the approach from~\cite{Nagarajan2010}.
Suppose for the sake of contradiction that $0< x(e) <\frac{1}{2\theta}$ for every $e \in E$.
Assign one token to every $e =uv \in E$,
putting $\theta x(e) >0$ ``end-tokens'' at each of $u$ and $v$, and $1-2\theta x(e)>0$ ``middle-tokens'' at $e$.
We will show that these tokens can be distributed such that every member of ${\cal L} \cup C$
gets one token, and the middle-tokens of the edges entering the maximal members of ${\cal L}$ are not assigned.
This gives the contradiction $|E|>|{\cal L}|+|C|$.

Every leaf $\hat{S}$ gets
$\theta x(\delta_E(\hat{S})) = \theta f(\hat{S}) \geq \theta$ end-tokens from edges in $E^+_S$.
Hence we have $\theta |{\cal E}|$ tokens at leaves.
By the assumption $(\theta-1)|{\cal E}| \geq |C|$, we have
$|{\cal E}| +|C| \leq \theta |{\cal E}|$,
so these tokens suffice to give one token to each member of ${\cal E} \cup C$.

Now let $\hat{S} \in {\cal L}$ be a non-leaf biset.
 Denote by $t(\hat{S})$ the amount of the following tokens:
 \begin{itemize}
  \item end-tokens of edges in $E^+_S$ at nodes owned by $\hat{S}$
	(these tokens always exist if $E^+_S \neq \emptyset$),
  \item middle-tokens of edges in $E^-_S$
	(these tokens always exist if $E^-_S \neq \emptyset$),
  \item end-tokens of edges in $E^-_S$ at nodes owned or shared by
	$\hat{S}$
	(these tokens exist if there exists an edge in $E^-_S \neq
	\emptyset$ that covers exactly one child of $\hat{S}$).
 \end{itemize}
 Note that $t(\hat{S})>0$, by the linear independence.
We claim that $t(\hat{S})$ is an integer, hence $t(\hat{S}) \geq 1$.
Let $E_1$ be the set of edges in $E^-_S$ that cover exactly one child of $\hat{S}$
and let $E_2=E^-_S\setminus E_1$ be the set of edges in $E^-_S$ that cover two distinct
children of $\hat{S}$. Let $E'$ be the set of edges in $E$ that cover both $\hat{S}$
and some child of $\hat{S}$. Then,
\begin{eqnarray*}
t(\hat{S}) & = & \theta x(E^+_S)+(|E_1|-\theta x(E_1))+(|E_2|-2\theta x(E_2)) \\
           & = & \theta [x(E^+_S)+x(E')]-\theta[x(E_1)+2x(E_2)+x(E')]+|E_1|+|E_2|    \\
           & = & \theta x(\delta_E(\hat{S}))-\theta \sum_{\hat{R} \in {\cal C}_S}
                 x(\delta_E(\hat{R}))+|E^-_S|
             =   \theta \left(f(\hat{S})-\sum_{\hat{R} \in {\cal C}_S}
			 f(\hat{R})\right)+|E^-_S| \ .
\end{eqnarray*}

 To each non-leaf biset $\hat{S} \in {\cal L}$,
 we distribute tokens counted in $t(\hat{S})$.
 Consequently, every members of ${\cal L} \cup C$
gets one token, and the middle-tokens of the edges entering
 the maximal members of ${\cal L}$ are not assigned.
  We note that any tokens are not counted more than once because 
 $E^+_S \cap E^+_{S'}=\emptyset$ and
 $E^-_S \cap E^-_{S'}=\emptyset$ hold
 for any distinct bisets $\hat{S},\hat{S}' \in \cal L$.
This gives the contradiction $|E|>|{\cal L}|+|C|$.
\end{proof}

We note that if $C=\emptyset$, i.e., if there are no degree bounds, then
the same proof applies for $\theta=1$ to show that
any extreme point of $P(f,b,E)$ has an edge $e \in E$ with $x(e) \geq 1/2$.
This coincides with the simple proof of the result of \cite{Fleischer2006}
given in \cite{Nagarajan2010}.

\begin{lemma} \label{l:beta}
If $|C| > (\alpha/2-1) |{\cal E}|$, then
there is $e \in E$ with $x(e) \geq 1/\alpha$ or
there is $v \in C$ with $|\delta_E(v)| \leq \beta$,
where $\beta =\left\lceil 4\frac{\gamma+2}{\alpha-2}\right\rceil + 5$.
\end{lemma}
\begin{proof}
Assume for a contradiction that
$0< x(e)<1/\alpha$ for every $e \in E$ and that
$|\delta_E(v)| \geq \beta+1$ for every $v \in C$.
We give one token to each end-node of every edge in $E$.
We will show that these tokens can be distributed such that every member of ${\cal L} \cup C$
gets two tokens, and each maximal member of ${\cal L}$ gets four tokens leading to the contradiction that $|E| > |{\cal L}| + |C|$.

The amount of tokens at each $v \in C$ is at least $\beta+1$.
Hence we have at least $(\beta+1)|C|$ tokens at the nodes in $C$.
From these tokens, we give four tokens to every leaf and $2(\Delta_{\cal L}(v)+2)$ tokens
to every $v \in C$.
This is possible by \eqref{e:delta} of Lemma~\ref{l:count} and the assumption 
$|C| >(\alpha/2-1) |{\cal E}|$, as we verify below.
\begin{multline*}
4|{\cal E}|+2\sum_{v \in C} (\Delta_{\cal L}(v)+2)  \leq
4|{\cal E}|+2(\gamma|{\cal E}|+3|C|) 
= 2|{\cal E}|(\gamma+2)+6|C| \\
= |C| \left(6 + 2 \frac{|{\cal E}|}{|C|} (\gamma+2)\right) 
< |C| \left(6 + 2 (\gamma+2) \frac{2}{\alpha -2}\right)
\leq |C| (1 + \beta).
\end{multline*}

 Every $v \in C$ will keep two tokens. From the remaining $2(\Delta_{\cal L}(v)+1)$
 tokens, $v$ will give two tokens to every biset that owns or shares $v$.
 Now we discuss the tokens given to bisets in $\cal L$.
 We show that the tokens can be rearranged as claimed by induction on the
 height of the forest corresponding to $\cal L$.
 
 Let $\hat{S} \in {\cal L}$ be a biset that is not a leaf.
 By the induction hypothesis, we assume that each descendant
 of $\hat{S}$ has at least two tokens, and each child of $\hat{S}$ has
 four tokens.
 We move two tokens from each child to $\hat{S}$.
If $\hat{S}$ has at least two children then we are done.
If $\hat{S}$ has one child and owns or shares a node $v \in C$,
 then $\hat{S}$ gets two tokens from its child and two tokens from $v$.
 Let us consider the other case (i.e., $\hat{S}$ has one child and owns or
 shares no node in $C$).
In this case, each edge in 
$E^+_S \cup E^-_S$ has one end-node that is owned or shared by
 $\hat{S}$, and this end-node is not contained by $C$.
 We give the tokens of these end-nodes to $\hat{S}$.
Note that $|E^+_S \cup E^-_S|\geq 2$, by linear independence and the integrality of $f$.
Hence $\hat{S}$ gets two tokens from its child and two tokens from
 end-nodes of edges in
$E^+_S \cup E^-_S$.
 Consequently,
 we can always give four tokens to $\hat{S}$, keeping two tokens for
 each descendant of $\hat{S}$.
\end{proof}

Applying Lemma~\ref{l:theta} with $\theta=\lfloor \alpha/2 \rfloor$,
we get that if $(\lfloor \alpha/2 \rfloor-1)|{\cal E}| \geq |C|$,
and in particular if $(\alpha/2-1)|{\cal E}| \geq |C|$,
then there is $e \in E$ with $x(e) \geq \frac{1}{2\lfloor \alpha/2 \rfloor} \geq \frac{1}{\alpha}$.
Together with Lemma~\ref{l:beta} this implies Theorem~\ref{t:ElC''}.

\subsection{Degree approximation only (Part (ii) of Theorem~\ref{t:ElC''})}

We call a biset in $\Lset$ {\em strictly black} if it owns a node in
$C$, {\em black} if one of its descendents is strictly black (i.e., its
inner-part contains a node in $C$), and 
{\em white} otherwise (i.e., its inner-part contains no node in $C$).
Let $\Eset_{\rm b}$ and $\Eset_{\rm w}$ denote the family of strictly
black bisets and white bisets in $\Eset$, respectively.

\begin{lemma}
 If $|\Eset| \leq (\gamma+4)|C|$, then there is $e \in E$ with $x(e) \geq
 1/2$, or there is $v \in C$ with $|\delta_E(v)| \leq 1.5\gamma^2+7.5\gamma+16$.
\end{lemma}
\begin{proof}
Assume for a contradiction that $0 < x(e) < 1/2$ for every $e \in E$ and 
$|\delta_E(v)| \geq 1.5\gamma^2+7.5\gamma+17$ for every $v \in C$.
Identifying a node $v \in C$ as a biset $(\{v\},\{v\})$,
we regard $\Lset \cup C$ as a biset family.
$\Lset \cup C$ may not be strongly laminar, but it is laminar.
Therefore we can define the inclusion order on $\Lset \cup C$.

 We assign two tokens to every edge in $E$, putting one end-token at
 each of its end-nodes.
 We will show that these
 tokens can be distributed such that every member of $\Lset \cup C$
 gets two tokens, and an extra token remains. This gives the
 contradiction that $|E| > |\Lset|+|C|$.

 Let $e=uv \in E$.
 Note that there always exists a biset $\hat{X} \in \Lset \cup C$ such that $e\in
 \delta_E(\hat{X})$. 
Suppose that $\hat{X}$ is a minimal one
 among such bisets.
 Without loss of generality, let $u \in X$. 
 We give the end-token of $e$ at $u$  to $\hat{X}$.
 If there also exists a biset $\hat{Y} \in \Lset \cup C$ such that 
 $e \in \delta_E(\hat{Y})$ and $v \in Y$, then we 
 give the end-token of $e$ at $v$ to
  the minimal such biset $\hat{Y}$.
 Otherwise, the end-token of $e$ at $v$ is given to the minimal biset
 $\hat{X}'$ such that $\hat{X} \subset \hat{X}'$ and $e \not\in \delta_E(\hat{X}')$.

 Since bisets in $\Eset_{\rm w}$ and nodes in $C$ are leaves of
 $\Lset \cup C$, they obtain one token from each edge incident to them after this distribution.
 Hence each biset $\hat{S} \in \Eset_{\rm w}$ has three
 tokens and each node $v \in C$ has $|\delta_E(v)|$ tokens.
 We make each $v \in C$ keep only two tokens, 
 return $1/3$ tokens to each edge in $\delta_E(v)$, and release the
 other tokens.
 Then the total number of released tokens is
 \[
 \sum_{v \in C}\left(\frac{2}{3}|\delta_E(v)|-2\right) 
 >
 |C|(\gamma^2+5\gamma+9)
 \geq (1+\gamma)|\Eset|+5|C|,
 \]
 where the first inequality follows from $|\delta_E(v)| >
 1.5\gamma^2+7.5\gamma+16.5$, $v \in C$ and 
 the last one follows from $|\Eset| \leq (\gamma+4)|C|$.
 We redistribute these tokens to bisets as follows:
 \begin{itemize}
  \item one token is given to each biset in $\Eset_{\rm w}$,
  \item four tokens are given to each biset in $\Eset_{\rm b}$,
  \item if $v \in C$ is shared by a biset $\hat{X}$,
	then one token is given to $\hat{X}$,
  \item if $v \in C$ is owned by a biset $\hat{X}$
	and $v$ is shared by no biset in $\Lset$,
	then two tokens are given to $\hat{X}$,
  \item if $v \in C$ is owned by a biset $\hat{X}$
	and $v$ is shared by some bisets in $\Lset$,
	then one token is given to $\hat{X}$.
 \end{itemize}
 If a biset owns or shares more than one node in $C$,
 it obtains tokens from each of those nodes in $C$ following the last
 three rules.
 This redistribution is possible because the number of tokens we need is 
 \begin{eqnarray*}
  |\Eset_{\rm w}| + 4|\Eset_{\rm b}|+\sum_{v \in C}(1+\max\{1,\Delta_{\Lset}(v)\})
  & \leq & (|\Eset|-|C|) + 4|C| + |C|+\gamma|\Eset| +|C|\\
   &=& (1+\gamma)|\Eset| + 5|C|,
 \end{eqnarray*}
 where the above inequality follows from
 $|\Eset_{\rm w}| + |\Eset_{\rm b}|=|\Eset|$,
$|\Eset_{\rm b}| \leq |C|$, and \eqref{e:delta}.

 Now all tokens given to bisets in $\Eset$ and nodes in $C$ have been redistributed
 such that
\begin{itemize}
 \item each node in $C$ has two tokens,
 \item each biset in $\Eset$ has four tokens,
 \item each pair of $v\in C$ and $e \in \delta_{E}(v)$ has $1/3$ tokens,
 \item each biset $\hat{X} \in \Lset$ has at least one token from each owning
       node in $C$,
       and one token from each sharing node in $C$.
       If a node $v$ owned by $\hat{X}$ is shared by no biset, then
       $\hat{X}$ has two tokens from $v$.
\end{itemize}

 Let $\hat{S} \in \Lset$, and let $\Lset'$ be the family of $\hat{S}$
 and its proper descendants.
 In what follows, we make
 each biset in $\Lset'$ receive at least two tokens,
 and $\hat{S}$ receive four tokens.
 For this, we redistribute tokens that were given to
 the bisets in $\Lset'$,
 and those kept by pairs of edge $e$ and its end-node $v \in C$
 such that $e$ is incident to a biset in $\Lset'$
 and $v$ is shared or owned by this biset.
 We prove that this redistribution is possible by induction on the
 height of the tree defined from $\Lset'$.
 If the height is one, the claim follows from
 that each biset in $\Eset$ has four tokens.
 Hence let us consider the case where the height of the tree is more
 than one.

 By the induction hypothesis, we can assume that each descendant
 has at least two tokens, and each child of $\hat{S}$ has four tokens.
 $\hat{S}$ can obtain two tokens from each of its child.
 Thus
 $\hat{S}$ can collect four tokens in each of the following cases:
 \begin{itemize}
  \item  $\hat{S}$ has more than one child;
  \item  $\hat{S}$ owns a node in $C$ that is shared by no biset in $\Lset$;
  \item  $\hat{S}$ owns or shares at least two nodes in $C$.
 \end{itemize}
 In the rest, we discuss the other case, and show that 
 $\hat{S}$ collects at least two tokens in addition to the tokens given from the child.
 Let $\hat{Y}$ be the child of $\hat{S}$.

 By the linear independence, $|E_S^+\cup E_S^-| \geq 2$ always holds, 
 and 
 $|E_S^+\cup E_S^-| \geq 3$ holds when either $E_S^+$ or $E_S^-$ is empty
 by the assumption that $x(e) < 1/2$, $e\in E$.
 Let us discuss how many tokens are given to $\hat{S}$ from the end-nodes of edges in $E_S^+\cup E_S^-$.
 
 Let $e=uv \in E_{S}^+$.
 Without loss of generality, we let $u \in S$ and $v \in V \setminus
 S^+$.
 Notice that $\hat{S}$ owns $u$.
 Hence if $u \in C$, then $\hat{S}$
 receives one or two tokens from $u$.
 If $u \not\in C$,
 $\hat{S}$ obtains the end-token of $e$ at $u$.
 Let $e'=u'v' \in E_{S}^-$. We let $u' \in S$ be the end-node 
 which is within $Y$, and consequently
 $v' \in S^+\setminus Y^+$.
 By the strongly laminarity of $\Lset$,
 no biset $\hat{X} \in \Lset$ with $e' \in
 \delta_E(\hat{X})$ contains
 $v'$ in its inner-part.
 This implies that, if $v' \not\in C$,
 $\hat{S}$ obtains the end-token of $e'$ at $v'$.
 If $v' \in C$, then $\hat{S}$ owns or shares $v'$,
 and obtains one or two tokens from $v'$.
 Summing up, if there is more than one node such as $u$ or $v'$,
 then we are done.
 Even if there exists exactly one such node, $\hat{S}$ receives at least
 two tokens unless the node is in $C$ and is shared by some bisets.

 Consider the case of $E^-_S \neq \emptyset \neq E^+_S$.
 We define $e=uv$ and $e'=u'v'$ as above.
 The above discussion shows that $\hat{S}$ receives two tokens unless $u=v' \in C$.
 If $u=v' \in C$, this node is not contained by $Y^+$,
 and hence no biset in $\Lset$ shares it.
 This means that $\hat{S}$
 always receives at least two tokens in this case.
 
 Next, consider the case of $E_S^+=\emptyset$.
 Then $|E_S^-| \geq 3$, 
 and all edges in $E_S^-$ are incident to the same
 node $v'$ in $C \cap (S^+ \setminus Y^+)$.
 Each of
 the edges in $E_S^-$ has $1/3$ tokens corresponding to $v'$.
 $\hat{S}$ collects
 one token from these edges, and another token from $v'$.
 Therefore, we are done.
 The claim is proven similarly when $E_S^-=\emptyset$.
\end{proof}

We next discuss the case where $|\Eset| \geq (\gamma+4)|C|$, and prove
the following lemma.

\begin{lemma}\label{lem.noedge-cost.smalle}
 If $x$ is maximal in $P(f,b,E)$ and every edge in $E$ is incident to some node in $B$ and
 $|\Eset| \geq (\gamma+4)|C|$, then there is $e \in E$ with $x(e) \geq
 1/2$.
\end{lemma}

Under the assumption in Lemma~\ref{lem.noedge-cost.smalle},
each edge $e \in E$ is incident to a node in 
$\{v \in B \colon x(\delta_E(v))=b(v)\}$ since otherwise we can increase
$x^*(e)$.
Since Lemma~\ref{l:rank} holds for arbitrary 
$\mathcal{L}$ and $C$ that satisfy the conditions described in the lemma,
we can define $C$ as an inclusion-wise maximal subset of $\{v \in B \colon
x(\delta_E(v))=b(v)\}$ such that the vectors in $\chi_E(C)$ are linearly independent.
If $C$ contains no end-node of $e \in E$, then 
the incidence vector of $\delta_E(v)$ defined from an end-node $v$ of $e$
is linearly independent from those in $\chi_E(C)$.
Since this contradicts the maximality of $C$, we can observe that each
edge $e \in E$ is incident to at least one node in $C$.

We again count bisets in $\Lset$ and nodes in $C$ for proving 
Lemma~\ref{lem.noedge-cost.smalle}, but the way of distributing tokens
is different here.
Let $e=uv \in E$.
By the assumption, at least one of the end-nodes of $e$ is in $C$.
If $C$ contains both end-nodes of $e$, then we assign no token to $e$.
If $C$ contains exactly one end-node, say $v$, of $e$, then we assign
one token. This token will be given to a biset in $\Lset$ as follows.
If $\Lset$ contains a biset $\hat{S}$ such that $e \in
\delta_E(\hat{S})$ and $u \in S$, then the token is given to such a
minimal biset.
 If there exists no such bisets and $\Lset$ contains a biset $\hat{X}$
 such that $e \in \delta_E(\hat{X})$ and $v \in X$,
       then the token is given to the minimal biset in $\{\hat{Y} \in
       \Lset \colon \hat{X} \subset \hat{Y}, u \in Y^+\}$.
Since the total number of tokens is at most $|E|$, it suffices to show
that an extra token remains after redistributing tokens so that each
biset in $\Lset$ and each node in $C$ owns one token.

Let $\hat{S} \in \Eset_{\rm w}$. 
Since $x(e) < 1/2$ for each $e \in E$, $|\delta_E(\hat{S})| \geq
2f(\hat{S})+1$. Since $S$ contains no nodes in $C$, 
each edge in $\delta_E(\hat{S})$ gives a token to $\hat{S}$.
Thus $\hat{S}$ has
$2f(\hat{S})+1 \geq f(\hat{S})+2$ tokens.
We make each $\hat{S} \in \Eset_{\rm w}$ release one token.
Then the number of released tokens is at least
$|\Eset_{\rm w}| \geq |\Eset|-|C| \geq (\gamma+3)|C|$.
Recall that the number of strictly black bisets is at most $|C|$.
We redistribute the released token to the nodes in $C$ and the strictly black
bisets so that
each $v \in C$ has one token,
and each strictly black biset has $\gamma+2$ tokens.
Note that each $\hat{S} \in \Eset_{\rm w}$ still has at least 
$f(\hat{S})+1$ tokens after this redistribution.

We first count tokens in a tree
which consists of only white bisets.

\begin{lemma}\label{lem.whitetree}
 Let $\hat{R} \in \Lset$ be a white biset, and 
 $\Lset'=\{\hat{S} \in \Lset \colon \hat{S} \subseteq \hat{R}\}$.
 We can distribute tokens owned by bisets in $\Lset'$
 so that each biset in $\Lset'$ has at least one token, and 
$\hat{R}$ has at least $1+f(\hat{R})$ tokens when $0 < x(e) < 1/2$ for each
 $e \in E$.
\end{lemma}
\begin{proof}
 We prove this by induction on the height of the tree representing $\Lset'$.
 If the height is one, then $\Lset'=\{\hat{R}\}$ and $\hat{R} \in
 \Eset_{\rm w}$. Thus the lemma follows in this case.
 
 Assume that the height is at least two.
 Applying the induction hypothesis to the trees rooted at the children
 of $\hat{R}$,
 we can allocate tokens so that 
 each biset below the children of $\hat{R}$ has one token, and
 each child $\hat{S}$ of $\hat{R}$ has $1+f(\hat{S})$ tokens.
 We can move $\sum_{\hat{S} \in \mathcal{C}_{R}} f(\hat{S})$ tokens from the children to
 $\hat{R}$.

 If $\sum_{\hat{S} \in \mathcal{C}_R} f(\hat{S}) > f(\hat{R})$, then we are done.
 Hence consider the other case.
 When $\sum_{\hat{S} \in \mathcal{C}_R} f(\hat{S}) = f(\hat{R})$,
 $|E_R^+| \geq 1$ holds by the linear
 independence.
  When 
 $\sum_{\hat{S} \in \mathcal{C}_R} f(\hat{S}) < f(\hat{R})$,
 $|E_R^+| \geq 1+ 2(f(\hat{R}) - \sum_{\hat{S}\in \mathcal{C}_R} f(\hat{S}))$ 
 holds by $x^*(e) < 1/2$, $e \in E$.
 In either case, 
 $|E^+_R| \geq  1+ f(\hat{R}) - \sum_{\hat{S} \in \mathcal{C}_R}
 f(\hat{S})$.
$\hat{R}$ is given a token from each $e\in E_R^+$ because
 $e$ has an end-node $v \in R$ such that $\hat{R}$ is a minimal biset
 with $e \in \delta_E(\hat{R})$ and $v \in R$, and $v\not\in C$ by $R\cap C=\emptyset$.
 Thus $\hat{R}$ has already owned $1+ f(\hat{R}) - \sum_{\hat{S} \in \mathcal{C}_R}
 f(\hat{S})$.
 With the tokens from the children, $\hat{R}$ obtains $1+f(\hat{R})$ tokens.
\end{proof}

We next give a token distribution scheme for trees in which the maximal bisets
are black. Together with Lemma~\ref{lem.whitetree}, this finishes the
proof of Lemma~\ref{lem.noedge-cost.smalle}.

\begin{lemma}\label{lem.blacktree}
 Let $\hat{R} \in \Lset$ be a black biset, and 
 $\Lset'=\{\hat{S} \in \Lset \colon \hat{S} \subseteq \hat{R}\}$.
 We can distribute tokens owned by bisets in $\Lset'$
 so that each biset in $\Lset'$ has at least one token, and 
$\hat{R}$ has at least $2$ tokens when $0 < x(e) < 1/2$ for each
 $e \in E$.
\end{lemma}
\begin{proof}
We show how to rearrange the tokens so that
each biset in $\Lset'$  obtains at least one token,
and $\hat{R}$ obtains at least $2+\gamma-|\Gamma(\hat{R})\cap C| \geq 2$ tokens.
Our proof is by induction on the height of the tree.
If the height is one, then the claim holds because it consists of a
strictly black biset.
Hence suppose that the height is at least two.

Let $\mathcal{B}$ be the set of black children of
$\hat{R}$,
and $\mathcal{W}$ 
be the set of white children of $\hat{R}$.
Apply the induction hypothesis to the subtrees rooted at the black children, 
and Lemma~\ref{lem.whitetree} to the subtrees rooted at the
white children. 
Then each biset below the children has one token, each $\hat{X} \in \mathcal{B}$ has
$2+\gamma-|\Gamma(\hat{X})\cap C|$ tokens, and
each $\hat{Y} \in \mathcal{W}$ has $1+f(\hat{Y}) \geq 2$ tokens.
If $\mathcal{B}=\emptyset$, then $\hat{R}$ is strictly black, and it has already given
$\gamma+2$ tokens. Since this finishes the claim, suppose that
 $\mathcal{B} \neq \emptyset$.
Since each child of $\hat{R}$ needs only one token, 
we can move extra tokens from the children to $\hat{R}$.
The number of tokens $\hat{R}$ obtains is at least
\begin{equation}\label{eq.tokens-subgraph}
|\mathcal{W}|+ \sum_{\hat{X} \in \mathcal{B}} (1+\gamma-|\Gamma(\hat{X})\cap C|).
\end{equation}

Let $\hat{S}$ be an arbitrary biset in $\mathcal{B}$.
A node $v\in \Gamma(\hat{S}) \cap C$ is either in $R$ or
$\Gamma(\hat{R})$.
If $v$ is in $R$, then we are done because $\hat{R}$ is a strictly black biset that owns $v$.
Therefore assume that each $v\in \Gamma(\hat{S}) \cap C$ is in
$\Gamma(\hat{R})$.
This means that $\Gamma(\hat{S}) \cap C \subseteq \Gamma(\hat{R}) \cap
C$, and hence $|\Gamma(\hat{S}) \cap C| \leq |\Gamma(\hat{R}) \cap
C|$.
Hence \eqref{eq.tokens-subgraph} is at least the required number of
tokens if $|\mathcal{W}| \geq 1$, if $|\mathcal{B}| \geq 2$, or if $\Gamma(\hat{S}) \cap C
\subset \Gamma(\hat{R}) \cap C$. 

Let $|\mathcal{B}|=1$, $|\mathcal{W}|=0$, and 
$\Gamma(\hat{S}) \cap C = \Gamma(\hat{R}) \cap C$.
Notice that $\hat{S}$ is the only child of $\hat{R}$ in this case.
It suffices to find one more token for $\hat{R}$.
The linear independence between $\chi_E(\hat{R})$ and
$\chi_E(\hat{S})$ implies that 
at least one of $E^+_R$ and $E^-_R$ is not empty.

Let $e \in E^+_R$.
Then $e$ has an end-node $v$ in $R \setminus S$. If $v \in C$, then
$\hat{R}$ is a strictly black biset that owns $v$, and hence $\hat{R}$
has the required number of tokens in this case as mentioned above.
If $v \not\in C$, then $e$ gives a token to $\hat{R}$.
Thus we are done when $E^+_R\neq \emptyset$.

Let $e'=u'v' \in E^-_R$.
Then $e'$ has an end-node, say $u'$, in $R^+ \setminus S^+$. 
If $u' \in C$, then
$\Gamma(\hat{S}) \cap C = \Gamma(\hat{R}) \cap C$
implies that $u' \in R$, and hence $\hat{R}$ is a strictly black biset.
Let $u' \not\in C$. Then $v' \in C$.
 $\Lset$ has no biset $\hat{Z}$ with $e' \in \delta_E(\hat{Z})$ and $u'
 \in Z$ by the strong laminarity of $\Lset$.
$e' \in \delta_E(\hat{S})$, $v' \in S\cap C$, 
and $\hat{R}$ is the
minimal biset such that $\hat{S} \subset \hat{R}$ and $u'\in R^+$.
Hence $e'$ gives one token to $\hat{R}$ in this case, which completes the proof.
\end{proof}

\section{Proof of Theorem~\ref{t:reductions}} \label{s:reductions}

Here we prove Theorem~\ref{t:reductions}, stating that {\SN} on undirected graphs
admits the following approximation ratios for any integer $\alpha \geq 1$.
\begin{itemize}
\item[\rm (i)]
$O(k^3\log |T|) \cdot (\alpha, \alpha b(v)+k/\alpha)$
for {\DB} {\NC} {\SN}.
\item[\rm (ii)]
$O(k\log k) \cdot (\alpha, \alpha b(v)+k/\alpha)$ for {\DB} {\sf Rooted} {\SN}.
\item[\rm (iii)]
$\frac{1}{\epsilon}O(k \log^2 k) \cdot (\alpha, \alpha b(v)+k/\alpha)$
for {\DB} {\SkCS} with $k \leq (1-\epsilon)|T|$ and $0<\epsilon <1$.
\end{itemize}

Part~(i) follows from Theorem~\ref{t:ElC} and the decomposition
of {\NC} {\SN} into $O(k^3 \log |T|)$ instances of {\ElC} {\SN}
due to Chuzhoy and Khanna~\cite{Chuzhoy2009}.

For proving (ii), we need to explain the algorithm of \cite{Nutov2009b}
for {\sf Rooted} {\SN} without degree-bounds.
By augmentation version, we denote instances of the problem
in which $G$ contains a subgraph $J$ of zero edge cost such that $\kappa(s,v) \geq r(s,v)-1$ for every
$v \in T$.
In \cite{Nutov2009b} it is shown that the augmentation version can be decomposed into
$O(k)$ instances of {\DB} {\fCS} with skew supermodular $f$.
The algorithm for the general version has $k$ iterations.
At iteration $\ell$, one adds to $J$ an edge set that increases the connectivity
by one for each node $v$ such that $\kappa(s,v)=r(s,v)-k+\ell -1$.
After iteration $\ell$ we have $\kappa(s,v) \geq r(s,v)-k+\ell$, hence after $k$ iterations the
solution becomes feasible.
In \cite{Nutov2009b} it is shown that
if the augmentation version admits an algorithm that computes a solution of cost at most $\alpha$ times
the optimal value of the corresponding biset LP relaxation,
then the general version admits ratio $O(\alpha \log k)$.
This is because if $x$ is a feasible solution
to LP relaxation derived from an instance of {\sf Rooted} {\SN},
then $\frac{x}{k-\ell+1}$ is feasible to the LP relaxation derived from the augmentation version.
For the case with degree-bounds, we proceed in the same way.
When we solve an augmentation version instance,
the degree-bounds $b'$ are defined by $b'(v)=\lceil \frac{b(v)}{k-\ell+1} \rceil$ for $v \in B$.
Then we claim that if there exists an
$(\alpha,\beta(b(v)))$-approximation algorithm for
{\DB} {\fCS} with skew supermodular $f$,
then {\DB} {\sf Rooted} {\SN} admits ratio
$O(k \log k) \cdot (\alpha,\beta(b(v)))$.
This and Theorem~\ref{t:ElC'} prove (ii).

We prove (iii).
In the augmentation version of {\SkCS}, the goal is to increase the connectivity between
the terminals from $k-1$ to $k$, namely,
$G$ contains a subgraph $J$ of zero edge cost such that $\kappa(u,v) \geq k-1$ for all $u,v \in T$.
We use a result of \cite{Nutov2011} that the augmentation version of
{\SkCS} with $k \leq (1-\epsilon)|T|$ is decomposed into $\frac{1}{\epsilon}O(\log k)$
instances of augmentation versions of {\sf Rooted} {\SN} with $r(s,v)=k$ for all $v \in T$.
To solve the general version of {\SkCS}, we repeatedly solve $k$ augmentation versions,
at iteration $\ell$ increasing the connectivity between the nodes in $T$ from $\ell-1$ to $\ell$.
As in the rooted case, if $x$ is a feasible solution
to LP relaxation derived from an instance of {\SkCS},
then $\frac{x}{k-\ell+1}$ is feasible to the LP relaxation derived from the augmentation version.
Hence if the augmentation version admits an algorithm that computes a solution of cost at most $\alpha$ times
the optimal value of the corresponding biset LP relaxation,
then the general version admits ratio $O(\alpha \log k)$.
This extends to the degree bounded setting, if at iteration $\ell$ we scale the degree bounds
to $b'(v)=\lceil \frac{b(v)}{k-\ell+1} \rceil$ for each $v \in B$.
Then we claim that if the augmentation version of {\DB} {\sf Rooted} {\SN} admits ratio
$(\alpha, \beta(b(v)))$ then {\DB} {\SkCS} admits ratio $\frac{1}{\epsilon}O(\log^2 k) \cdot (\alpha, \beta(b(v)))$.
By \cite{Nutov2009b} the augmentation version of {\sf Rooted} {\SN} can be decomposed into
$O(k)$ instances of {\fCS} with skew supermodular $f$, and this also extends to the degree bounded setting.
Overall, we obtain that {\DB} {\fCS} with skew supermodular $f$ admits ratio
$(\alpha, \beta(b(v)))$ then {\DB} {\SkCS} admits ratio $\frac{1}{\epsilon}O(k\log^2 k) \cdot (\alpha, \beta(b(v)))$.
This and Theorem~\ref{t:ElC'} prove (iii).

Part~(iii) does not mention the case $k > (1-\epsilon)|T|$.
In this case, compute a minimum cost set of $k$ internally disjoint $(u,v)$-paths for each
pair of $u,v \in T$, and define a solution as the union of these paths.
Note that the $k$ internally disjoint $(u,v)$-paths can be computed by a minimum cost flow algorithm.
The edge cost of this solution is $O(k^2)$ times the optimal, and
the degree of each node is $O(k^2)$ because $|T|=O(k)$.

\section{Proof of Theorem~\ref{t:new-reduction}} \label{s:new-reduction}

We need to describe the algorithm of \cite{Nutov2012} for {\DB} {\kCS}.
The algorithm uses the following procedure due to Khuller and Raghavachari \cite{Khuller1996},
that is also used in the next section.

\vspace*{0.2cm}

\begin{center}
\fbox{
\begin{minipage}{0.960\textwidth}
\begin{description}\setlength{\itemsep}{0pt}
\item[Procedure {\sc External $k$-Out-connectivity}]
\item[Input:] A graph $G=(V,E)$, an integer $k$, and $R\subseteq V$ with $|R|=k$.
\item[Output:] A subgraph $J$ of $G$.
\item[Step 1:]
Let $G'$ be obtained from $G$ by adding a new node $s$ and
all edges between $s$ and $R$, of cost zero each.
\item[Step 2:]
Compute a $k$-outconnected from $s$ spanning subgraph $J'$ of $G'$.
\item[Step 3:]
Return $J=(J' \setminus \{s\})$.
\end{description}
\end{minipage}
}
\end{center}

\vspace*{0.2cm}

Assume that {\DB} {\kOS} admits an $(\alpha,\beta(b(v)))$-approximation algorithm.
For undirected graphs, the algorithm of \cite{Nutov2012} is as follows.

\vspace*{0.2cm}

\begin{center}
\fbox{
\begin{minipage}{0.960\textwidth}
\begin{description}\setlength{\itemsep}{0pt}
\item[Algorithm {\sc Degree Bounded $k$-Connectivity}]
\item[Step 1:]
Apply Procedure {\sc External $k$-Out-connectivity},
where $J'$ is computed using the $(\alpha,\beta(b(v)))$-approximation algorithm for {\DB} {\kOS}
with degree bounds $b'(v)=b(v)+1$ if $v \in R$ and $b'(v)=b(v)$ otherwise.
\item[Step 2:]
Let $F$ be a set of edges on $V$ such that $J \cup F$ is $k$-connected.
\item[Step 3:]
For every $ut \in F$ compute a minimum-cost inclusion-minimal edge-set $I_{ut} \subseteq E \setminus J$
such that $J \cup I_{ut}$ contains $k$ internally disjoint $ut$-paths.
\item[Step 4:]
Return $J \cup I$, where $I=\cup_{ut \in F}I_{ut}$.
\end{description}
\end{minipage}
}
\end{center}

\vspace*{0.2cm}

In the case of directed graphs, Procedure {\sc External $k$-Out-connectivity}
computes a subgraph $J'=J^- \cup J^+$, where:
$J^+$ is $k$-outconnected from $s$ and is computed by the $(\alpha, \beta(b(v)))$-approximation algorithm,
while $J^-$ is a minimum cost subgraph which is $k$-inconnected to $s$, namely,
$\kappa_{J^-}(v,s) \geq k$ for each $v \in V\setminus \{s\}$.

\begin{lemma} [\cite{Nutov2012}] \label{l:reduction-core}
Algorithm {\sc Degree Bounded $k$-Connectivity} has ratio
$(\alpha+|F|,\beta(b(v))+2|F|+kd/2)$ for undirected graphs, and
$(\alpha+1+|F|,\beta(b(v))+k+|F|+kd/2)$ for digraphs,
where $F$ is the edge set computed at Step~2, $d=\max_{v \in V}
 |\delta_F(v)|$ for undirected graphs, and $d=\max_{v \in
 V}|\delta^+_F(v)|$ for digraphs.
\end{lemma}

Let $F$ be an edge set that satisfies the condition at Step~2 of Algorithm
{\sc Degree Bounded $k$-Connectivity}.
In \cite{Kortsarz2003,Nutov2012}, it is shown that there exists 
such an edge set $F$ on $R$.
Moreover, if $F$ is an inclusion-minimal such edge set, then
$F$ is a forest in the undirected case, and 
$F$ contains no alternating cycle
(a cycle such that every two successive arcs have opposite directions)
in the directed case.
The latter property is a known consequence from the undirected and directed Critical Cycle Theorems of
Mader \cite{Mader1972,Mader1985}.
In addition, the latter property implies $|F|\leq 2|R|-1$.
This can be seen as follows. We make a copy $R'$ of $R$, and replace each arc $uv \in F$ by an
undirected edge $uv'$, where $v'$ is the copy of $v$. Then we obtain an
undirected edge set on $R \cup R'$. This edge set is a forest
if $F$ contains no alternating cycle.
Thus $|F| \leq 2|R|-1$.

Therefore, we can find in polynomial time $F$ with
$|F| \leq |R|-1=k-1$ in the undirected case, and
$|F| \leq 2|R|-1=2k-1$ in the directed case.
We improve this by showing that $F$ as above can be converted in polynomial time into an edge set $F'$
such that $|F'|\leq |F|$, $\max_{v \in V}|\delta_{F'}(v)|=O(\sqrt{k})$
for undirected graphs, and 
$\max_{v \in V}|\delta^+_{F'}(v)|=O(\sqrt{k})$
for digraphs.

\subsection{Undirected graphs}

We start by proving the following.

\begin{lemma}\label{l:transform-undirected}
Let $G=(V,E \cup F)$ be a simple $k$-connected undirected graph such that $|\delta_E(v)| \geq p$ for all $v \in V$,
and $G \setminus \{e'\}$ is not $k$-connected for each $e' \in F$.
Let $d=\max\limits_{v \in V} |\delta_F(v)| \geq 4$,
let $u \in V$ with $|\delta_F(u)|=d$, and let $e=ut \in \delta_F(u)$.
Suppose that for every $v \in V$ with $|\delta_F(v)| \leq d-2$ the graph
$G \setminus \{ut\} \cup \{vt\}$ is not $k$-connected. Then
 $d(d+p-k-2)\leq 3p-k+1$.
\end{lemma}
\begin{proof}
$F$ is a forest by Mader's Critical Cycle Theorem for undirected graphs.
Consider the graph $G \setminus \{e\}$ and the biset family
\[
{\cal F}=\{\hat{S} \in {\cal V}\colon  u \in S, t \in V \setminus S^+,
         |\Gamma(\hat{S})|=k-1, \delta_{E \cup F}(\hat{S})=\{e\}\}  \ .
\]
${\cal F}$ is a ring biset family, namely, that
(i) the intersection of the inner parts of the members 
of ${\cal F}$ is non-empty, and 
(ii) $\hat{X} \cap \hat{Y}, \hat{X} \cup \hat{Y} \in {\cal F}$ for any 
$\hat{X},\hat{Y} \in {\cal F}$.
 Indeed, (i) is obvious because
 the inner part of each biset in ${\cal F}$ includes
 $t$.
 (ii) follows from the fact that the functions $|\Gamma(\cdot)|$ and
 $|\delta_{E \cup F}|$ satisfy the submodular inequality
 (the reverse of \eqref{e:super}), and the biset family $\{\hat{S} \in
 {\cal V} \colon u \in S, t \in V\setminus S^+\}$ is closed under union and intersection.
 This implies that ${\cal F}$ has a unique minimal member $\hat{S}$,
and that for every $v \in S$ the graph $G \setminus \{ut\} \cup \{vt\}$ is $k$-connected.
Thus $|\delta_F(v)| \geq d-1$ for every $v \in S$, implying that
$|\delta_{E \cup F}(v)| \geq p+d-1$.

Let $K=\Gamma(\hat{S})$, so $|K|=k-1$.
Let $I$ be the set of edges in $F$ with at least one end-node in $S$.
 Every edge in $I \setminus \{e\}$ has both of its end-nodes in $S \cup K$
 because $\delta_{E \cup F}(\hat{S})=\{e\}$.
In $G$, every $v \in S \setminus \{u\}$ has at least $p+d-1$ neighbors in $(S \setminus\{v\}) \cup K$, implying
$(|S|-1)+(k-1) \geq p+d-1$, so $p+d-k+1 \leq |S|$.
Since $I$ is a forest on a set $S \cup K \cup \{t\}$ of $|S|+k$ nodes, $|I| \leq	|S|+k-1$.
Let $\zeta_I(S)$ be the set of edges in $I$ with both end-nodes in $S$,
so $I$ is a disjoint union of $\delta_I(S)$ and $\zeta_I(S)$.
Hence $|\delta_I(S)|+|\zeta_I(S)|=|I| \leq |S|+k-1$.
On the other hand, $(d-1)|S| \leq \sum_{v \in S} |\delta_I(v)|=|\delta_I(S)|+2|\zeta_I(S)|$.
Summarizing, we have the following:
\begin{eqnarray}
 p+d-k+1                           & \leq & |S|,  \label{e:1} \\
|\delta_I(S)|+|\zeta_I(S)|  & \leq & |S|+k-1,   \label{e:2} \\
|\delta_I(S)|+2|\zeta_I(S)| & \geq & (d-1)|S|.  \label{e:3}
\end{eqnarray}
Subtracting (\ref{e:2}) from (\ref{e:3}) gives $|\zeta_I(S)| \geq (d-2)|S|-k+1$ and thus
$|\delta_I(S)| \leq 2k-2-(d-3)|S|$. Since $|\delta_I(S)| \geq 0$ we get
$(d-3)|S| \leq 2k-2$. Combining with (\ref{e:1}) we get
$$
p+d-k+1 \leq |S| \leq \frac{2k-2}{d-3}.
$$
Multiplying by $d-3$ and rearranging terms we obtain $d (d+p-k-2)\leq 3p-k+1$, as claimed.
\end{proof}

\begin{corollary}\label{cor.reduction-undirected}
Let $G=(V,E \cup F)$ be a simple $k$-connected undirected graph such that
$|\delta_E(v)| \geq k-1$ for all $v \in V$.
Then there exists a polynomial time algorithm that finds
a set $F'$ of edges on $V$ with $|F'| \leq |F|$ such that $G'=(V,E \cup F')$ is $k$-connected
and such that $|\delta_{F'}(v)| \leq \max\left\{3,\frac{3}{2}+\sqrt{2k+\frac{1}{4}}\right\}$ for all
$v \in V$.
\end{corollary}
\begin{proof}
 Let $u \in V$ be a node that maximizes $|\delta_F(u)|$.
 Lemma~\ref{l:transform-undirected} with $p=k-1$ implies that
 if $|\delta_F(u)|$ is larger than the
 required value, then
 we can replace an edge $ut \in \delta_F(u)$
 by another edge $vt$ such that $|\delta_F(v)|$ is at most $|\delta_F(u)|-2$,
  keeping the graph being  $k$-connected.
 By repeating this replacement, we can obtain a required edge set $F'$.
\end{proof}

The undirected part of Theorem~\ref{t:new-reduction} follows from
Lemma~\ref{l:reduction-core}, Corollary~\ref{cor.reduction-undirected},
and our ability to find in polynomial time an edge set $F$ with $|F|\leq k-1$
at Step~2 of Algorithm {\sc Degree Bounded $k$-Connectivity}.

\subsection{Digraphs}
We start by proving the directed counterpart of Lemma~\ref{l:transform-undirected}.

\begin{lemma}\label{l:transform-directed}
Let $G=(V,E \cup F)$ be a simple $k$-connected digraph such that 
$|\delta^+_E(v)|\geq p$ for all $v \in V$,
and $G\setminus \{e'\}$ is not k-connected for each $e' \in F$.
Let $d=\max\limits_{v \in V} |\delta^+_F(v)| \geq 4$,
let $u \in V$ with $|\delta^+_F(u)|=d$ and let $e=ut \in \delta^+_F(u)$.
Suppose that for every $v \in V$ with $|\delta^+_F(v)| \leq d-2$ the graph
$G \setminus \{e\} \cup \{vt\}$ is not $k$-connected. Then
 $d(d+p-k-2)\leq 3p-k+2$.
\end{lemma}
\begin{proof}
Consider the graph $G \setminus \{e\}$ and the biset family
\[
{\cal F}=\{\hat{S} \in {\cal V}\colon  u \in S, t \in V \setminus S^+,
         |\Gamma(\hat{S})|=k-1, \delta^+_{E\cup F}(\hat{S})=\{e\}\}  \ .
\]
As in Lemma~\ref{l:transform-undirected}, it can be shown that ${\cal F}$ is a ring biset family,
so ${\cal F}$ has a unique minimal member $\hat{S}$,
and that for every $v \in S$ the graph $G \setminus \{ut\} \cup \{vt\}$ is $k$-connected.
Thus $|\delta^+_F(v)| \geq d-1$ for every $v \in S$,
implying that $|\delta^+_{E \cup F}(v)| \geq p+d-1$ for every $v \in S$.

Let $K=\Gamma(\hat{S})$,
so $|K|=k-1$.
Let $I$ be the set of arcs in $F$ with tail in $S$.
Every edge in $I \setminus \{e\}$ has its head in $S \cup K$.
In $G$, every $v \in S\setminus \{u\}$ has at least $p+d-1$ neighbors in $(S \setminus\{v\}) \cup K$,
 implying $(|S|-1)+(k-1) \geq p+d-1$, so $d-k+p+1 \leq |S|$.
Since $F$ has no alternating cycle,
$I$ is an arc set without alternating cycle on a set $S \cup K \cup \{t\}$ of $|S|+k$ nodes.
This implies that $|I| \leq2(|S|+k)-1$.
On the other hand, $(d-1)|S| \leq \sum_{v \in S} |\delta^+_I(v)| =|I|$.
Summarizing, we have the following:
\begin{eqnarray}
d-k+p+1                          & \leq & |S|,   \label{e:1'} \\
|I| & \leq & 2(|S|+k)-1,  \label{e:2'} \\
|I| & \geq & (d-1)|S|.    \label{e:3'}
\end{eqnarray}
From (\ref{e:2'}) and (\ref{e:3'}) we get $(d-1)|S| \leq 2(|S|+k)-1$ so
$|S|(d-3) \leq 2k-1$. Combining with (\ref{e:1'}) we get
$$
d-k+p+1 \leq |S| \leq \frac{2k-1}{d-3}.
$$
Multiplying by $d-3$ and rearranging terms we obtain
$d(d+p-k-2)\leq 3p-k+2$, as claimed.
\end{proof}

\begin{corollary}\label{cor.reduction-directed}
Let $G=(V,E \cup F)$ be a simple $k$-connected digraph such that
$\min_{v \in V} |\delta^-_E(v)| \geq k-1$
and  $\min_{v \in V} |\delta^+_E(v)| \geq k-1$.
Then there exists a polynomial time algorithm that finds a set $F'$ of arcs on $V$ with $|F'| \leq |F|$
such that $G'=(V,E \cup F')$ is $k$-connected
and such that $|\delta^+_{F'}(v)|$ and $|\delta^-_{F'}(v)|$
are both at most
$\max\left\{3,1.5+\sqrt{2k+1.25}\right\}$ for all $v \in V$.
\end{corollary}
\begin{proof}
Let $u \in V$ be a node that maximizes $|\delta^+_F(u)|$.
Lemma~\ref{l:transform-directed} with $p=k-1$ implies that
if $|\delta^+_F(u)|$ is larger than the
required value, then
we can replace an edge $ut \in \delta^+_F(u)$
by another edge $vt$ such that $|\delta^+_F(v)|$ is at most $|\delta^+_F(u)|-2$,
keeping the graph being  $k$-connected.
By repeating this replacement, we can obtain $F''$ that satisfies the
conditions on connectivity and out-degree.
Notice that $|\delta^-_{F''}(v)|=|\delta^-_{F}(v)|$ for all $v \in V$.
Similarly we can decrease the in-degree of a node in $V$
if it is larger than the required value,
by applying Lemma~\ref{l:transform-directed} to the graph obtained
by reversing the directions of all arcs.
This gives the required edge set $F'$.
\end{proof}

The directed part of Theorem~\ref{t:new-reduction} follows from
Lemma~\ref{l:reduction-core}, Corollary~\ref{cor.reduction-directed},
and our ability to find in polynomial time an edge set $F$ with $|F|\leq 2k-1$
at Step~2 of Algorithm {\sc Degree Bounded $k$-Connectivity}.

\section{Proof of Theorem~\ref{t:kCS}} \label{s:kCS}

First, we overview the algorithm of Cheriyan and
V\'egh~\cite{CVegh2012},
and show that it can be extended to the degree-bounded setting.
Then, we improve the bound on the number of nodes.

\subsection{Extension to the degree-bounded setting}
\label{s:CV}

Define a biset function $f^k\colon  \Vset \rightarrow \Zset$
as
\[
 f^k(\hat{S}) =
 \begin{cases}\displaystyle
 k -|\Gamma(\hat{S})| & \mbox{if $S \neq \emptyset$ and $S^+ \neq V$} \\
  0                   & \mbox{otherwise.}
 \end{cases}
\]
By the node-connectivity version of Menger's Theorem,
an undirected graph $(V,F)$ is $k$-connected if and only if
$|\delta_F(\hat{S})| \geq f^k(\hat{S})$ for each $\hat{S} \in \Vset$.
Now suppose that our goal is to augment a given graph $(V,J)$ by a minimum-cost edge set $F$
such that $(V,J \cup F)$ is $k$-connected.
A natural LP relaxation for this problem is as follows (see \cite{Frank1995}).
\begin{equation} \label{e:kconn}
\tau^*=\min\left\{\sum_{e \in E} c(e) x(e) \colon x(\delta_E(\hat{S}))
	    \geq f^k_J(\hat{S}) \text{ for } \forall \hat{S} \in \Vset ,
	    0 \leq x(e) \leq 1 \text{ for } \forall e \in E\right\}
\end{equation}
where $f^k_J(\hat{S})=f^k(\hat{S})-|\delta_J(\hat{S})|$ is the residual biset function of $f^k$. 
We will denote ${\cal S}_J=\{\hat{S} \colon f^k_J(\hat{S})>0\}$. 
Recall also that we denote $\gamma=\max_{f(\hat{S})>0} |\Gamma(\hat{S})|$, and note that $\gamma \leq k-1$
for $f=f^k_J$. In what follows, we assume that $k \geq 2$.

Two bisets $\hat{X}$ and $\hat{Y}$
{\em cross}      if $X \cap Y \neq \emptyset$        and $X^+ \cup Y^+ \neq V$, and
{\em nega-cross} if $X \setminus Y^+ \neq \emptyset$ and $Y \setminus X^+ \neq \emptyset$.
A biset function $f$ is {\em crossing supermodular} if any
$\hat{X},\hat{Y} \in \Vset$ that cross satisfy the supermodular inequality \eqref{e:super}.
$f$ is {\em symmetric} if $f(S,S^+)=f(V \setminus S^+,V \setminus S)$ for any biset $\hat{S}=(S,S^+) \in {\cal V}$.
It is known that the function $f^k_J$ is crossing supermodular and symmetric for any edge set $J$.


A biset family ${\cal F}$ is {\em independence-free} if any $\hat{X},\hat{Y} \in {\cal F}$ cross or nega-cross.
We say that a biset function $f$ is independence-free if the family $\{\hat{S} \in \Vset \colon f(\hat{S})>0\}$ 
is independence-free, and
that an edge set $J$ is independence-free if the biset function $f^k_J$ is independence-free
(namely, if the family ${\cal S}_J$ is independence-free).
The idea of Cheriyan and V\'egh is to find a ``cheap''
independence-free edge set $J$.
They also showed that
if $J$ is independence-free, the iterative rounding algorithm of
\cite{Fleischer2006}
for skew supermodular biset functions
computes an edge set $F \subseteq E$ with $c(F) \leq 2\tau^*$ such that $(V,J \cup F)$ is $k$-connected.

The step for finding such $J$
is based on the following statement.

\begin{lemma}[\cite{CVegh2012}] \label{l:CVegh}
 Let $J'$ be an undirected graph
on a node set $V \cup\{s\}$
 such that $J'$ is $k$-outconnected from $s$. Let $R$
be the set of neighbors of $s$ in $J'$, and let $J=J' \setminus\{s\}$.
Let $U=\bigcup\{S \colon \hat{S} \in {\cal S}_J, |S| \leq k-1\}$.
Then $|U| \leq |R|k^2 (k-1)$.
Furthermore, if $|V| \geq |U|+k$, then there exists a polynomial time algorithm that
given an edge set $E$ on $V$ with costs returns one of the following:
\begin{itemize}
\item[{\rm (i)}]
An edge set $F \subseteq E$ with $c(F) \leq 2\tau^*$ such that $J \cup F$ is $k$-connected.
\item[{\rm (ii)}]
The set $U$.
\end{itemize}
\end{lemma}

The algorithm constructs an edge set that corresponds to
$J'$ in Lemma~\ref{l:CVegh}
by applying Procedure {\sc External $k$-Out-connectivity} from the beginning of 
Section~\ref{s:new-reduction}.
Frank and Tardos \cite{Frank1989} gave a polynomial-time algorithm for computing a subgraph that is spanning $k$-outconnected from a root node in a directed graph.
This implies a $2$-approximation algorithm for the same problem in undirected graphs \cite{Khuller1996};
then it computes a subgraph of cost at most $2\tau^*$.
Procedure {\sc External $k$-Out-connectivity} uses this
$2$-approximation algorithm in its Step~2.

The algorithm of Cheriyan and V\'egh has four steps.
At every step, a certain edge set of cost at most $2\tau^*$ is computed.
If the algorithm terminates at Step~2, then it returns the union of the edge-sets computed at Steps~1 and 2,
of overall cost at most $4\tau^*$.
Else, the algorithm returns the union of the edge-sets computed at Steps~1,~3, and 4, 
of overall cost at most $6\tau^*$.

\vspace*{0.2cm}

\begin{center}
\fbox{
\begin{minipage}{0.960\textwidth}
\begin{description}\setlength{\itemsep}{0pt}
\item[Algorithm of Cheriyan and V\'egh]
\item[Step 1:]
Compute a subgraph $J_{CV}$ of $G$ 
by applying Procedure {\sc External $k$-Out-connectivity}
for some $R \subseteq V$ with $|R|=k$.
\item[Step 2:]
Apply the algorithm from Lemma~\ref{l:CVegh}.
If the algorithm returns an edge set $F$ as in Lemma~\ref{l:CVegh}(i)
then return $J_{CV} \cup F$ and STOP.
\item[Step 3:]
If the algorithm from Lemma~\ref{l:CVegh} returns $U_{CV}$,
then apply Procedure {\sc External $k$-Out-connectivity} for some $R
	   \subseteq V \setminus U_{CV}$
	   with $|R|=k$,
and add the computed edge set 
 to $J_{CV}$.
(Then, the graph $J_{CV}$ is independence-free.)
\item[Step 4:]
Apply the iterative rounding algorithm of \cite{Fleischer2006}
to compute an edge set $F \subseteq E$ 
such that $J_{CV} \cup F$ is $k$-connected.
\end{description}
\end{minipage}
}
\end{center}

\vspace*{0.2cm}
Step~3 of this algorithm needs a condition $|V|\geq |U_{CV}|+k$
to find $R$. Lemma~\ref{l:CVegh} shows that $|U_{CV}|\leq k^3(k-1)$, and hence
$n\geq k^3(k-1)+k$ suffices for guarantee the condition.

The algorithm can be extended to the degree bounded setting as follows.
At Steps 1 and 3, we apply Procedure {\sc External $k$-Out-connectivity}
with degree bounds $b'(v)=b(v)+1$ if $v \in R$ and $b'(v)=b(v)$ otherwise,
using our algorithm for undirected {\DB} {\kOS}.
At Steps 2 and 4, we use our algorithm for {\DB} {\fCS} with skew supermodular $f$.
If $|V| \geq |U|+k$, then following~\cite{CVegh2012},
we can design a polynomial time algorithm that returns either the set $U$,
or an edge set $F \subseteq E$ such that $J \cup F$ is $k$-connected,
within the same ratio as our algorithm for {\DB} {\fCS} with skew
supermodular $f$ (with $\gamma=k-1$).
These give $(12,8b(v)+O(k))$-approximation algorithm for
 {\DB} {\sf Undirected} {\kCS}.



\subsection{Improving the bound on the number of nodes}

In the rest, we improve the bound on $n$ as described in Theorem~\ref{t:kCS}.
For the degree-bounded setting,
we simply improve the statement $|U| \leq |R| k^2(k-1)$ $(=k^3(k-1))$ in
Lemma~\ref{l:CVegh}
to $|U| \leq 2k(k-1)(k-0.5)$.
Since the last claim in Lemma~\ref{l:CVegh} requires $|V|\geq |U|+k$,
this improvement proves our claim.

For the setting without degree bounds, we slightly modify the algorithm
of Cheriyan and V\'egh.  
Aulet\-ta~et~al.\ \cite{Auletta1999} gave a procedure
for computing a spanning
subgraph $J'$ of an undirected graph $G$ such that $J'$ is $k$-outconnected from some node
$r$ in $G$, $|\delta_{J'}(r)|=k$, and $c(J') \leq 2 \tau^*$ (this procedure
does not apply in the degree bounded setting, if we care about the
cost).  We apply this procedure at Step~1, instead of the {\sc External
$k$-Out-connectivity} procedure to obtain a subgraph $J'$ (instead of
$J_{CV}$).
Define $U'$ as $\bigcup\{S\colon  \hat{S} \in {\cal S}_{J'},|S| \leq
k-1\}$.
We apply Steps~2, 3, and 4 with $J'$ and $U'$ instead of $J_{CV}$ and $U_{CV}$.
We will prove that $|U'| +k \leq k(k-1)(k-1.5)+k$ holds,
and hence we can weaken the assumption on $|V|$ to $|V| \geq k(k-1)(k-1.5)+k$.

Now we describe the main result in this subsection.

\begin{lemma} \label{l:main}
Let $J'$ be an undirected graph such that $J'$ is $k$-outconnected from some node $s$ with 
$|\delta_{J'}(s)|=k$ and let $J=J' \setminus\{s\}$. Let
$U' =\bigcup\{S\colon  \hat{S} \in {\cal S}_{J'},|S| \leq p\}$
and
$U=\bigcup\{S\colon  \hat{S} \in {\cal S}_J,|S| \leq p\}$.
Then $|U'| \leq pk(k-1.5)$ and $|U| \leq 2pk(k-0.5)$.
In particular, for $p=k-1$, 
$$|U'| +k \leq k(k-1)(k-1.5)+k \ \ \ \mbox{and} \ \ \ |U| +k \leq 2k(k-1)(k-0.5)+k \ .$$
\end{lemma}

Let us say that a biset family ${\cal F}$ is {\em weakly nega-uncrossable}
if for any $\hat{X},\hat{Y} \in {\cal F}$ with
$X \setminus Y^+,Y\setminus X^+ \neq \emptyset$, one of the bisets
$\hat{X} \setminus \hat{Y}, \hat{Y} \setminus \hat{X}$ is in ${\cal F}$.
We have the following lemma on
weakly nega-uncrossable families.

\begin{lemma}\label{l:posi-uncrossable}
If $f$ is crossing supermodular 
and symmetric, then the biset family ${\cal F}=\{\hat{S} \colon f(\hat{S})>0\}$ is
weakly nega-uncrossable.
\end{lemma}
\begin{proof}
Suppose that $\hat{X},\hat{Y} \in {\cal F}$
satisfy
$X \setminus Y^+ \neq \emptyset$ and $Y\setminus X^+ \neq \emptyset$.
Let $\hat{X}'$ be $(V\setminus X^+,V\setminus X)$.
Then, $f(\hat{X}')+f(\hat{Y})=f(\hat{X})+f(\hat{Y}) >0$ holds because $f$ is symmetric.
 Since $\hat{X}'$ and $\hat{Y}$ are crossing, we have
$f(\hat{X}')+f(\hat{Y}) \leq f(\hat{X}' \cap \hat{Y}) + f(\hat{X}'\cup \hat{Y})$
 by the crossing supermodularity of $f$.
 Note that
 $\hat{X}' \cap \hat{Y}=\hat{Y}\setminus \hat{X}$ and
  $\hat{X}' \cup \hat{Y}=\hat{X}\setminus \hat{Y}$.
Thus combining these gives
$f(\hat{X}\setminus \hat{Y})+f(\hat{Y}\setminus \hat{X}) >0$, and hence $f(\hat{X} \setminus \hat{Y}) >0$ or
$f(\hat{Y}\setminus \hat{X}) >0$ holds.
\end{proof}

Two bisets $\hat{X}$ and $\hat{Y}$ are {\em strongly disjoint} if
$\hat{X} \setminus \hat{Y}=\hat{X}$ or $\hat{Y} \setminus \hat{X}=\hat{Y}$
(note that this is equivalent to
$\hat{X} \setminus \hat{Y}=\hat{X}$ and $\hat{Y} \setminus \hat{X}=\hat{Y}$;
in particular, $X \subseteq V \setminus Y^+$ and $Y \subseteq V \setminus X^+$).
Given a biset family ${\cal F}$, let $\nu_{\cal F}$ denote the
maximum number of pairwise strongly disjoint bisets in ${\cal F}$.

\begin{lemma} \label{l:posi}
Let ${\cal F}$ be a weakly nega-uncrossable biset family.
Denote $p=\max_{{\hat S} \in {\cal F}}|S|$
and
$\gamma=\max_{\hat{S} \in {\cal F}} |\Gamma(\hat{S})|$.
Then $\left|\bigcup_{{\hat S} \in {\cal F}} S \right| \leq p (2\gamma+1) \nu_{\cal F}$.
\end{lemma}
\begin{proof}
Let ${\cal F'}=\{\hat{S}_1,\hat{S}_2,\ldots,\hat{S}_\ell\}$ be a minimum size sub-family of ${\cal F}$ such that
$\bigcup_{i=1}^\ell S_i = \bigcup_{{\hat S} \in {\cal F}} S$.
We prove that $|{\cal F}'| \leq (2\gamma+1) \nu_{\cal F}$.
For every $\hat{S}_i \in {\cal F}'$, there is $v_i \in S_i$ such that $v_i \notin S_j$
for every $j \neq i$. Among all bisets in ${\cal F}$ that are contained in $\hat{S}_i$
and that includes $v_i$ in its inner-part, let $\hat{C}_i$ be a minimal one.
Since ${\cal F}$ is weakly nega-uncrossable,
the minimality of $\hat{C}_i$ implies that one of the following must
 hold for any distinct $\hat{C}_i$ and $\hat{C}_j$:
\begin{itemize}
\item
$v_i \in \Gamma(\hat{C}_j)$ or $v_j \in \Gamma(\hat{C}_i)$;
\item
$\hat{C}_i=\hat{C}_i \setminus \hat{C}_j$ or $\hat{C}_j=\hat{C}_j \setminus \hat{C}_i$,
namely, $\hat{C}_i, \hat{C}_j$ are strongly disjoint.
\end{itemize}
Construct an auxiliary directed graph ${\cal J}$ on node set
${\cal C}=\{\hat{C}_1,\hat{C}_2,\ldots,\hat{C}_\ell\}$.
Add an arc $\hat{C}_i\hat{C}_j$ if $v_i \in \Gamma(\hat{C}_j)$.
The in-degree in ${\cal J}$ of a node $\hat{C}_i$ is at most
$|\Gamma(\hat{C}_i)| \leq \gamma$.
This implies that every subgraph of the underlying graph
of ${\cal J}$ has a node of degree $\leq 2\gamma$.
A graph is $d$-degenerate if every subgraph
of it has a node of degree $\leq d$. It is known that any $d$-degenerate graph is $(d+1)$-colorable.
Hence ${\cal J}$ is $(2\gamma+1)$-colorable, so its node set can be partitioned into
$2\gamma+1$ independent sets.
The members of each independent set are pairwise strongly disjoint, hence their number is at most $\nu_{\cal C}$.
Consequently, $\ell \leq (2\gamma+1) \nu_{\cal C} \leq (2\gamma+1) \nu_{\cal F}$, as claimed.
\end{proof}

To prove Lemma~\ref{l:main}, we need the following lemma, which is also
used for proving Lemma~\ref{l:CVegh} in \cite{CVegh2012}.

\begin{lemma}[\cite{Khuller1996}] \label{l:R}
Let $J'$ be an undirected graph such that $J'$ is $k$-outconnected from some node $s$, let $R$
be the set of neighbors of $s$ in $J'$, and let $J=J' \setminus\{s\}$.
Then $S \cap R \neq \emptyset$ for any $\hat{S} \in {\cal S}_J$.
\end{lemma}

The following version of Lemma~\ref{l:R} is proved in \cite{Auletta1999};
we provide a proof-sketch for completeness of exposition.

\begin{lemma} \label{l:R'}
Let $J'$ be an undirected graph such that $J'$ is $k$-outconnected from some node $s$ and let $R$
be the set of neighbors of $s$ in $J'$.
Then $s \in \Gamma(\hat{S})$ and $|S \cap R| \geq 2$ for any $\hat{S} \in {\cal S}_{J'}$.
Hence $\nu_{{\cal S}_{J'}} \leq \left\lfloor |R|/2 \right\rfloor$.
\end{lemma}
\begin{proof}
Let $\hat{S} \in {\cal S}_{J'}$.
If $s \notin \Gamma(\hat{S})$ then $s \in S$ or $s \in V \setminus S^+$.
Since $S \neq \emptyset$ and $S^+ \neq V$, and since $J'$ is $k$-outconnected from $s$,
we easily obtain a contradiction to Menger's Theorem.
 We prove that $|S \cap R| \geq 2$. Let $v \in S$,
 $C=\delta_{J'}(\hat{S}) \cup \Gamma(\hat{S})$,
 and 
$\ell=|\delta_{J'}(\hat{S})| +|\Gamma(\hat{S})|$ ($\leq k-1$).
Consider a set of $k$ internally disjoint paths from $s$ to $v$ in $J'$.
At most $|S \cap R|$ of these paths may not contain a member in $C$.
This implies that each of the other at least $k-|S \cap R|$ paths
must contain an element from $C \setminus \{s\}$.
Hence $\ell-1 \geq k-|S \cap R|$. This implies $|S \cap R| \geq k-(\ell-1) \geq 2$.
\end{proof}

Now let us prove Lemma~\ref{l:main}.
Let ${\cal S}_J^p=\{\hat{S}\colon \hat{S} \in {\cal S}_J, |S| \leq p\}$.
Since ${\cal S}_{J'}$ and ${\cal S}_J$ are weakly nega-uncrossable by
Lemma~\ref{l:posi-uncrossable}, so are
${\cal S}_{J'}^p$ and ${\cal S}_J^p$.
We have $|R|\leq k$ because $\delta_{J'}(s)=k$.
Thus Lemmas~\ref{l:R} and \ref{l:R'} imply
$\nu_{{\cal S}_J^p} \leq k$ and $\nu_{{\cal S}_{J'}^p} \leq \lfloor k/2 \rfloor$, respectively.
Applying Lemma~\ref{l:posi} to ${\cal S}_J^p$, we obtain
\[
 |U|   \leq  p(2(k-1)+1) \nu_{{\cal S}_J^p}     \leq p(2k-1)k = 2pk(k-0.5).
\] 
When we apply Lemma~\ref{l:posi} to ${\cal S}_{J'}^p$,
we may assume $\gamma \leq k-2$.
This is because, by Lemma \ref{l:R'}, $s \in \Gamma(\hat{S})$ holds for any $\hat{S} \in {\cal S}_{J'}$,
and hence we can apply Lemma~\ref{l:posi} after removing $s$ from the boundary of every biset in ${\cal S}^p_{J'}$.
Hence we have
\begin{eqnarray*}
|U'|  \leq  p(2(k-2)+1) \nu_{{\cal S}_{J'}^p} \leq p(2k-3)\lfloor k/2 \rfloor  \leq
              pk(k-1.5)
\end{eqnarray*}
This concludes the proof of Lemma~\ref{l:main}.

The claim on the degree-bounded setting of Theorem~\ref{t:kCS} is
immediate from Section~\ref{s:CV} and Lemma~\ref{l:main}.
As for the setting without degree bounds,
we have to verify that the above modification of the algorithm makes no
effect on the claim. More specifically, we need to show that the algorithm
claimed in Lemma~\ref{l:CVegh} exists even if $U$ is replaced by $U'$
and $J$ is replaced by $J'$.
We note that the proof in \cite{CVegh2012} still holds even after
the replacement. Since it is not our main focus,
we leave it to the readers.

\section{Conclusion}\label{sec.conclusion}
We have presented iterative rounding algorithms and decomposition results for various {\DB} {\SN} problems.
We introduced several novel ideas in the field, which may be applicable
also to {\NC} {\SN} problems without degree bounds.
We believe that this is an important direction for future work.

\section*{Acknowledgments}
A part of this work was done when the third author was a visiting professor at
RIMS, Kyoto University in Spring 2011 and
when the first
author was visiting Carnegie Mellon University in 2011-12, supported by
Kyoto University Foundation.
The first author was also supported by
Japan Society for the Promotion of Science (JSPS), Grants-in-Aid for
Young Scientists (B) 25730008.
The third author was supported in part by NSF grants CCF 1143998 and CCF 1218382.
The authors thank anonymous referees for their careful reading and
insightful comments.
The presentation of this paper was improved by their suggestions.

\end{document}